\documentclass[a4paper,11pt]{article}

\usepackage{tikz}
\usetikzlibrary{shadows,arrows,shapes,positioning,calc}
\tikzset{
ell/.style={draw,ellipse,minimum height=2em,minimum width=8em,align=center},
line/.style={-,line width=0.5pt},
dashline/.style={-,dashed,line width=0.5pt},
dot/.style = {draw,fill,circle,inner sep=0pt,outer sep=0pt,minimum size=2pt},
}

\usepackage{algorithm}
\usepackage{latexsym}

\usepackage{fullpage}
\usepackage{color}
\usepackage{graphicx}
\usepackage{listings}
\usepackage{url}

\newcommand{\ceiling}[1]{\lceil #1\rceil}

\newtheorem{proposition}{Proposition}
\newenvironment{proof}{\emph{Proof:}}{$\Box$\newline}

\title{An improved, easily computable combinatorial lower bound
  for weighted graph bipartitioning}
\author{Jesper Larsson Tr\"aff, Martin Wimmer\\
Vienna University of Technology (TU Wien)\\
Faculty of Informatics, Institute of Information Systems\\ 
Research Group Parallel Computing\\
Favoritenstrasse 16/184-5, 1040 Vienna, Austria\\
email: \texttt{\{traff,wimmer\}@par.tuwien.ac.at}
}

\begin{document}
\maketitle

\begin{abstract}
There has recently been much progress on exact algorithms for the
(un)weighted graph (bi)partitioning problem using branch-and-bound and
related methods.  In this note we present and improve an easily
computable, purely combinatorial lower bound for the weighted
bipartitioning problem. The bound is computable in $O(n\log n+m)$ time
steps for weighted graphs with $n$ vertices and $m$ edges. In the
branch-and-bound setting, the bound for each new subproblem can be
updated in $O(n+(m/n)\log n)$ time steps amortized over a series of $n$
branching steps; a rarely triggered tightening of the bound requires
search on the graph of unassigned vertices and can take from $O(n+m)$
to $O(nm+n^2\log n)$ steps depending on implementation and possible
bound quality. Representing a subproblem uses $O(n)$ space. Although
the bound is weak, we believe that it can be advantageous in a
parallel setting to be able to generate many subproblems fast,
possibly out-weighting the advantages of tighter, but much more
expensive (algebraic, spectral, flow) lower bounds.

We use a recent priority task-scheduling framework for giving a
parallel implementation, and show the relative improvements in bound
quality and solution speed by the different contributions of the lower
bound. A detailed comparison with standardized input graphs to other
lower bounds and frameworks is pending. Detailed investigations of
branching and subproblem selection rules are likewise not the focus
here, but various options are discussed.
\end{abstract}

\section{Introduction}

There has recently been much progress on the exact solution of graph
partitioning problems, see for instance the
survey~\cite{MeyerhenkeSanders13}, as well as on heuristics for graphs
with special structure like for instance road networks.  In
particular, Delling et
al.~\cite{DellingGoldbergRazenshteynWerneck11,DellingGoldbergRazenshteynWerneck12,DellingWerneck12}
investigate new combinatorial lower bounds (based on approximations of
maximum flow-minimum cut bounds) for the unweighted problem, and
perform computational studies within a parallel branch-and-bound
framework~\cite{BudiuDellingWerneck11}. Armbruster et al. study linear
and semidefinite programming
approaches~\cite{ArmbrusterFugenschuhHelmbergMartin12}, and present a
sequential computational study. Improved flow-based bounds were given
in~\cite{Sensen01}, also with a computational study.

This note investigates another, simple, combinatorial lower-bound
approach which applies to both the weighted and unweighted graph
(bi)partitioning problems. The basic lower bound was originally
proposed in the early 90ties~\cite{Traff91:or,Traff94:ppl} with some
later improvements~\cite{Traff96:slb}. We present proofs and further
improvements, and implementations within the parallel task-scheduling
framework Pheet\footnote{The framework with the implementations
  described in this note can be downloaded from \url{www.pheet.org}}
which has been extensively described in~\cite{Wimmer14:diss}. The
motivation for this bound is the belief that weaker, but more easily
computable bounds may be preferable for parallel branch-and-bound over
stronger but hard-to-compute bounds in order to keep a large number of
processing units (threads, processes, cores, processors, \ldots) busy
throughout the solution of the given partitioning problems. This was
observed in~\cite{Traff94:ppl};
and~\cite{BudiuDellingWerneck11,DellingGoldbergRazenshteynWerneck12}
  give similar motivations for their bounds.

\section{The graph partitioning problem}

Given a weighted, undirected graph $G=(V,E)$ with vertices (or nodes,
used synonymously) $V$ and edges $E$ with arbitrary (real or integer)
edge weights $w(u,v), (u,v)\in E$, the \emph{graph bipartitioning
  problem} is to find a partition (\emph{cut}) of $V$ into two subsets
$V_0$ and $V_1$ of given sizes $|V_0|=s_0$ and $|V_1|=s_1$ with
$s_0+s_1=n$ and $s_0>0, s_1>0$ having minimum cut weight $w(V_0,V_1)$
over all such partitions. The weight of a cut is defined by extension
of the weight function as
\begin{eqnarray*}
w(V_0,V_1) & = & \sum_{\{(u,v)\in E | u\in V_0,v\in V_1\}}w(u,v)
\end{eqnarray*}
for any two disjoint subsets $V_0\subset V$ and $V_1\subset V$. The
graph partitioning problem is NP-hard, see
e.g.~\cite{GareyJohnson79,GareyJohnsonStockmeyer76}. The natural (and
relevant) generalization of the problem to partitioning $V$ into $k,
k>2$ subsets $V_i$ with predefined sizes $|V_i|=s_i$ (or with
predefined total vertex costs) is not discussed here, but many of the
observations carry over to the $k$-partitioning problem also.

\section{Lower and upper bounds}

We solve the graph partitioning problem using \emph{branch-and-bound},
a standard, search based method~\cite{PapadimitriouSteiglitz82} which
is presumably well-suited to parallel implementation, see,
e.g.~\cite{CrainicLeCunRoucairol06,GendronCrainic94,Talbi06}. The
essential components of a branch-and-bound algorithm are the notions
of \emph{subproblem}, \emph{completion}, \emph{lower bound}, and
\emph{branching rule}. The lower bound provides for any subproblem a
bound on the cut value of any completion of the subproblem. As soon as
the lower bound for a subproblem is larger than or equal to some
current, best feasible solution (or \emph{upper bound}) the subproblem
can be discarded from further consideration since it can never lead to
a better solution.

Any partition of the vertex set $V$ into a pair of subsets $(V_0,V_1)$
that fulfills $|V_i|=s_i, i=0,1$ is a \emph{feasible solution} to the
graph partitioning problem.  A \emph{subproblem} is a pair $(U_0,U_1)$
of disjoint subsets of $V$ with with $|U_i|\leq s_i,i=0,1$
representing a partial assignment of vertices to either of the two
subsets. A \emph{completion} of a subproblem $(U_0,U_1)$ is a feasible
solution $(V_0,V_1)$ with $U_i\subseteq V_i,i=0,1$. Vertices of $G$ in
either of $U_i$ are said to be \emph{fixed}, otherwise \emph{free}.
The set of free nodes is thus $F=V\setminus (U_0 \cup U_1)$. The
\emph{branching rule} selects a free node $v\in F$ and creates two new
subproblems by extending either of the sets $U_i$ with $v$, such that
$(U_0\cup\{v\},U_1)$ and $(U_0,U_1\cup\{v\})$ will be the two new
subproblems to be considered. The branch-and-bound process starts from
an empty subproblem $(\emptyset,\emptyset)$, respectively, if $n$ is
even, from a subproblem $(\{u\},\emptyset)$ for some node $u$ in order to
avoid generating symmetric solutions.

Let $n=|V|$ and $m=|E|$. We assume that $G$ has no self-loops $(u,u)$;
such edges never contribute to a cut anyway. We also assume that edges
$(u,v)\in E$ have non-negative costs. For the implementation, we let
$V=\{0,\ldots,n-1\}$. We represent a subproblem $(U_0,U_1)$ by two
bitmaps $B_i$ of $n$ bits; bit $u$ of $B_i$ is set iff $u\in U_i,
i=0,1$. Furthermore, we also maintain a bitmap for the free vertices,
and in addition an array of free vertices with $f=|F|=n-|U_0|-|U_1|$
being the number of free vertices. The weighted input graph $G$ is
represented by an array of adjacency arrays. Note that each edge
$(u,v)\in E$ is present in the adjacency arrays of both node $u$ and
of node $v$. We also need for each edge $(u,v)$ in the $i$th position
of the adjacency array of $u$ the position $j$ of $u$ in the adjacency
array of $v$. Finally, we store the adjacency arrays in sorted,
non-decreasing weight order.

The lower bounds are based on the following simple observation.  Let
$(V_0,V_1)$ be a completion of a subproblem $(U_0,U_1)$. It holds that
\begin{eqnarray*}
w(V_0,V_1) & = & w(U_0,U_1) + 
w(V_0\setminus U_0,U_1) + w(U_0,V_1\setminus U_1) +
w(V_0\setminus U_0,V_1\setminus U_1)
\end{eqnarray*}
A lower bound for a subproblem $(U_0,U_1)$ is therefore given by the
cut between already assigned vertices in $U_0$ and $U_1$, $w(U_0,U_1)$,
plus a lower bound on the term $w(V_0\setminus
U_0,U_1)+w(U_0,V_1\setminus U_1)$, and finally a lower bound on the term
$w(V_0\setminus U_0,V_1\setminus U_1)$. The latter two contributions
can be treated independently.

\subsection{The lower bound: basic bound and rebalancing}
\label{sec:basicrebalance}

Let $v\in F$ be a free node in the subproblem $(U_0,U_1)$. Any
completion $(V_0,V_1)$ will have a contribution to the cut value from
$v$ of at least $\min(w(v,U_0),w(v,U_1))$, no matter whether $v$ is
eventually in $V_0$ or $V_1$. Namely, if $v$ is in $V_0$ all
edges from $v$ to nodes in $U_1$ will contribute to the cut, and
similarly if $v$ is in $V_1$. Thus, a trivial lower bound for
the term  $w(V_0\setminus U_0,U_1) + w(U_0,V_1\setminus U_1)$ is
\begin{eqnarray*}
B(U_0,U_1) & = & \sum_{v\in F}\min(w(v,U_0),w(v,U_1))
\end{eqnarray*}

Computing this bound from scratch takes $O(n+m)$ steps. If we maintain
for each (free) vertex $v$ the two values $D_i[v] = w(v,U_i)$ for
the cost of assigning $v$ to subset $V_{i-1}$, the lower bound can be
computed as $\sum_{v\in F}\min(D_0[v],D_1[v])$ in $O(f)$ time steps
where $f=|F|$ is the number of free nodes. When branching on node $v$
and $v$ is put into $V_i$, all values $D_{i}[u]$ where $(v,u)\in
E$ need to be increased by $w(v,u)$. This can be done in $O(\deg(v))$
steps.

The bound $B(U_0,U_1)$ does not take the cardinality constraints on
completions of $(U_0,U_1)$ into account. If, for instance,
$w(v,U_1)<w(v,U_0)$ for a large number of nodes, then $B(U_0,U_1)$ may
count too many vertices as having been assigned to subset $V_0$, and a
stronger bound could be obtained by counting some of these vertices as
assigned to $V_1$. Define $\delta(v) =
w(v,U_1)-w(v,U_0)=D_1[v]-D_0[v]$ as the \emph{potential free weight
  increase} of node $v$. If $\delta(v)>0$ vertex $v$ would tend to be
assigned to subset $V_1$, and there is a penalty of $\delta(v)$ of
assigning $v$ to $V_0$ instead; if $\delta(v)<0$ the lower bound
$B(U_0,U_1)$ would count $v$ as assigned to $V_0$, and there would be
a penalty of $-\delta(v)$ of instead assigning $v$ to $V_1$. Penalties
are the amounts of which the lower bound might be increased when the
cardinality of the sets $V_0\setminus U_0$ and $V_1\setminus U_1$ are
taken into account. 

Let $\delta_i, 0\leq i<f$ be the potential free weight increases in
sorted order, $\delta_i\leq \delta_{i+1}$ for $0\leq i<f-1$. Then the
lower bound can be strengthened by a \emph{rebalancing contribution}
\begin{eqnarray*}
R(U_0,U_1) & = & 
\sum_{i=0}^{f_0-1}\max(0,\delta_i)+\sum_{i=f_0}^{f_1-1}\max(0,-\delta_i)
\end{eqnarray*}
where $f_i=s_i-|U_i|, i=0,1$. This basic rebalancing bound was first
presented in~\cite{Traff91:or,Traff94:ppl}. Computing the rebalancing
contribution seems to require sorting of the $\delta(v), v\in F$
values and can be done easily in $O(f\log f)$ steps. Our
implementation computes the rebalancing bound in this
fashion. Maintaining the $\delta(v)$ values in a priority queue does
not improve complexity, since up to $f$ values have to be considered
in order, and each extract min operation takes logarithmic
time. However, if the $\delta(v)$ values are maintained in sorted
order, recomputation and sorting is necessary only for $\deg(v)$ nodes
when branching on vertex $v$. The full array of $f$ values can be
reestablished by merging. The complexity of the rebalancing steps is
hereby reduced to $O(f+\deg(v)\log\deg(v)))$ which is $O(n^2+m\log n)$
for the whole a series of at most $n$ branching steps.

\begin{proposition}
\label{prop:rebalancing}
For any given subproblem $(U_0,U_1)$ it holds that
\begin{eqnarray*}
B(U_0,U_1)+R(U_0,U_1) & \leq & 
w(V_0\setminus U_0,U_1)+w(U_0,V_1\setminus U_1)
\end{eqnarray*}
for any completion $(V_0,V_1)$. The bound is \emph{tight}: there is a
completion $(V_0,V_1)$ such that $B(U_0,U_1)+R(U_0,U_1) =
w(V_0\setminus U_0,V_1)+w(U_0,V_1\setminus U_1)$.
\end{proposition}

\begin{proof}
As argued above, $B(U_0,U_1)$ is a lower bound on $w(V_0\setminus
U_0,V_1)+w(U_0,V_1\setminus U_1)$ in any completion $(V_0,V_1)$: Each
assigned vertex will contribute a weight of either $w(v,U_1)$ or
$w(U_0,v)$. The crucial part is the rebalancing step. 

Let $(V_0,V_1)$ be a completion that minimizes $w(V_0\setminus
U_0,U_0)+w(U_0,V_1\setminus U_1)$. Pick any two nodes $u\in
V_0\setminus U_0$ and $v\in V_1\setminus U_1$. It must hold that
$w(u,U_1)+w(v,U_0) \leq w(u,U_0)+w(v,U_1)$ since otherwise the weight
$w(V_0\setminus U_0,U_1)+w(U_0,V_1\setminus U_1)$ could be reduced by
swapping $u$ and $v$. This implies that $\delta(u)\leq\delta(v)$.  Let
$f_i=|V_i\setminus U_i|, i=0,1$.  We now prove by induction on
$\min(f_0,f_1)$ that
\begin{eqnarray*}
B(U_0,U_1)+R(U_0,U_1) & = & 
w(V_0\setminus U_0,U_1)+w(U_0,V_1\setminus U_1)
\end{eqnarray*}
for such a completion. Assume first that either $f_0=0$ or $f_1=0$. If
$f_0=0$, all free vertices are assigned to $V_1$, and for each $v$ it
holds that $w(v,U_0) = \min(w(v,U_0),w(v,U_1))+\max(0,-\delta(v))$;
namely, if $w(v,U_0)>w(v,U_1)$ then $w(v,U_0) =
w(v,U_1)-(w(v,U_1)-w(v,U_0))=w(v,U_1)-\delta(v)$.  If instead $f_1=0$,
then it holds that
$w(v,U_1)=\min(w(v,U_0),w(v,U_1))+\max(0,\delta(v))$. Therefore, in
either case
\begin{eqnarray*}
w(V_0\setminus U_0,U_1)+w(U_0,V_1\setminus U_1) & = & 
\sum_{v\in F}\min(w(v,U_0),w(v,U_1)) + \\
& &
\sum_{i=0}^{f_0-1}\max(0,\delta_i)+\sum_{i=f_0}^{f_1-1}\max(0,-\delta_i)
\\
& = & B(U_0,U_1)+R(U_0,U_1)
\end{eqnarray*}

Now assume that $\min(f_0,f_1)>0$. Choose a vertex $u\in
V_1\setminus U_0$ that maximizes $\delta(u)$ over all such $u$, and a
vertex $v\in V_1\setminus U_1$ that minimizes $\delta(v)$ over all
such $v$. Recall that $\delta(u)\leq\delta(v)$.  The contribution of
$u$ to $w(V_0\setminus U_0,U_1)+w(U_0,V_1\setminus U_1)$ is
$w(u,U_1)$, which, if $\delta(u)\geq 0$ can be written as
$\min(w(u,U_0),w(u,U_1))+\delta(u) = w(u,U_1)$ since $w(u,U_0)\leq
w(u,U_1)$. The contribution of $u\in V_0\setminus U_0$ is therefore
$\min(w(u,U_0),w(u,U_1))+\max(0,\delta(u))$. Likewise, the
contribution of $v$ to $w(V_0\setminus U_0,U_1)+w(U_0,V_1\setminus
U_1)$ is $w(v,U_0)$ which, by a similar case analysis, can be written
as $\min(w(v,U_0),w(v,U_1))+\max(0,-\delta(v))$. We can now remove $u$
and $v$ from the set of free edges. The resulting completion
$(V_0\setminus\{u\},V_1\setminus\{v\})$ minimizes
$w(V_0\setminus\{u\}\setminus
U_0,U_1)+w(U_0,V_1\setminus\{v\}\setminus U_1)$, and $\delta(u')\leq
\delta(v')$ for all $u'\in V_0\setminus\{u\}\setminus U_0$ and $v'\in
V_1\setminus\{v\}\setminus U_1$. By the induction hypothesis

\begin{eqnarray*}
w(V_0\setminus\{u\}\setminus U_0,U_1)+w(U_0,V_1\setminus\{v\}\setminus
U_1) & = & \sum_{v\in F\setminus\{u,v\}}\min(w(v,U_0),w(v,U_1)) + \\
& & \sum_{i=0}^{f'_0-1}\max(0,\delta_i)+\sum_{i=f'_0}^{f'_1-1}\max(0,-\delta_i)
\end{eqnarray*}
where $f'_i=f_i-1$ is the size of the subsets with $u$ and $v$ removed.
Adding in the contribution from $u$ and $v$ establishes the lower
bound claim.
\end{proof}

To represent a subproblem, our implementation uses $O(n)$ for the
bitmaps of fixed and free nodes, the array of free nodes, and the
$D_i$ values for the free vertices.

\subsection{The lower bound: high-degree unassigned vertices}
\label{sec:highdegree}

The lower bound counts edges between fixed and free vertices in the
subproblem $(U_0,U_1)$ but is oblivious to contributions from edges
between free nodes. However, some of these edges inevitably contribute
to a lower bound on completions of $(U_0,U_1)$. Let wlog $U_0$ be the
subset with the largest number of nodes still to be assigned, that is
assume that $f_0\geq f_1$. Consider the graph $G'=(F,E')$ induced by
the set of free nodes $F$, and let $v\in F$. If the degree of $v$ in
$G'$ is larger than $f_0-1$ then there will be at least
$\deg'(v)-f_0+1$ edges out of $v$ in any cut of $F$ into subsets of
size $f_0$ and $f_1$. Here, $\deg'(v)$ denotes the degree of $v$ in
$G'$.  The smallest weight such edges are a lower bound for the
contribution of $v$. Let $T_i(v) =
\sum_{j=0}^{\max(0,\deg'(v)-f_i+1)}w_j(v)$, where $w_j$ is the $j$th
smallest weight of an edge in $G'$ adjacent to $v$. Summing these
contributions over all free nodes and dividing by two since both
vertices of a cut edge may have a lower bound contribution gives
\begin{eqnarray*}
T(U_0,U_1) & = & \sum_{v\in F}T_0(v)/2
\end{eqnarray*}
For integer edge weights, $\ceiling{\sum_{v\in F}T_0(v)/2}$ is still a
lower bound, as will follow from the argument below.  These
observations were first made in~\cite{Traff96:slb}. Also here
rebalancing can be applied. If there are more than $f_0$ free nodes of
high degree, the lower bound counts too many nodes as becoming
assigned to $V_0$. Let $\delta'(v) = T_1(v)-T_0(v)$ be the
\emph{penalty} of assigning $v$ to the larger subset of size $f_0$; if
$\delta'(v)>0$ there is a gain of assigning $v$ instead to the subset
of size $f_1$ (note that for all $v$, $\delta'(v)\geq 0$). Again, let
$\delta'_i$ be the penalties in sorted order. Then the rebalancing
contribution is
\begin{eqnarray*}
R'(U_0,U_1) & = & \sum_{i=f_0}^{f_1-1}\delta'_{i-f_0}/2
\end{eqnarray*}
Note that rebalancing gives a contribution only if the number of
high-degree vertices is larger than $f_0$.

\begin{proposition}
\label{prop:highdegree}
For any given subproblem $(U_0,U_1)$ it holds that
\begin{eqnarray*}
T(U_0,U_1)+R'(U_0,U_1) & \leq & w(V_0\setminus U_0,V_1\setminus U_1)
\end{eqnarray*}
for any completion $(V_0,V_1)$.
\end{proposition}

\begin{proof}
  We prove that $T(U_0,U_1)$ is a lower bound for the partitioning
  problem on the graph $G'$ induced by $F$. Let $(u,v)$ be an edge in
  some cut $(W_0,W_1)$ of $F$ fulfilling the constraints $|W_i|=f_i,
  i=0,1$. The proof is by induction on the number of cut edges $(u,v)$
  where either $u$ or $v$ is adjacent to a high-degree vertex with
  degree larger than $f_0-1$. Pick one such cut edge. Let $u$ be a
  high degree vertex, and assume that $(u,v)$ is the $i$th smallest
  edge adjacent to $u$ in the cut. Note that there must be at least
  $\deg'(u)-f_0+1$ edges adjacent to $u$ in any cut since $u$ is a
  high-degree vertex and $G'$ has no self-loops.  If $i<\deg(u)-f_0$,
  the lower bound has a contribution from the $i$th smallest edge
  $(u,v')$ of weight $w(u,v')$ with $w(u,v')\leq w(u,v)$ (note that
  $v'$ may be $v$). Similarly if $v$ is a high-degree vertex. The
  contribution to the lower bound for cut edge $(u,v)$ is
  $w(u,v')/2+w(v,u')/2$, which is at most $w(u,v)$, since
  $w(u,v')+w(u',v)\leq 2w(u,v)$. If $v$ is not a high-degree vertex,
  there is no contribution from $v$, but it still holds that
  $w(u,v')/2\leq w(u,v)$. The cut edge $(u,v)$ can therefore be
  removed from the cut, and the edges $(u,v')$ and $(v,u')$ (if
  present) from $G'$, and the claim now follows by induction. Proof
  that rebalancing still yields a lower bound can be done by
  induction, similarly to the proof of
  Proposition~\ref{prop:rebalancing}.
\end{proof}

\begin{figure}
\centering
\begin{tikzpicture}[scale=0.75]
\node[dot](a) {};
\node[dot,right=2cm of a] (b){};
\node[dot,below=2cm of b] (c){};
\node[dot,left= 2cm of c] (d){};
\node[dot] at ($(d)!0.25!(a)$)(f){};
\node[dot] at ($(d)!0.5!(a)$)(g){};
\node[dot] at ($(d)!0.75!(a)$)(h){};
\node[dot] at ($(b)!0.5!(c)$)(k){};
\node[dot] at ($(b)!0.75!(c)$)(l){};
\node[ell,rotate=90] at ($(b)!0.5!(c)$) {};
\node[ell,rotate=90] at ($(d)!0.5!(a)$) {};
\draw[line] (a) -- (b);
\draw[dashline] (a) -- (c);
\draw[dashline] (a) -- (k);
\draw[dashline] (a) -- (l);
\draw[dashline] (b) -- (f);
\draw[dashline] (b) -- (g);
\draw[dashline] (b) -- (h);
\node[above] at ($(a)!0.5!(b)$) {$3$};
\node[above] at ($(a)!0.9!(k)$) {$2$};
\node[above] at ($(f)!0.1!(b)$) {$2$};
\node[left=3mm] at (a) {$u$};
\node[right=3mm] at (b) {$v$};
\node[right=3mm] at (k) {$v'$};
\node[left=3mm] at (f) {$u'$};
\end{tikzpicture}
\caption{The base case for the high-degree lower bound proof. The cut
  edge $(u,v)$ for two high-degree nodes $u$ and $v$ is shown as a
  heavy line, the smaller lower bound edges $(u,v')$ and $(v,u')$ as
  dotted lines; $u'$ and $v'$ may or may not be high-degree
  nodes. Indeed $(w(u,v')+w(v,u'))/2 = 2\leq w(u,v)=3$, but it does
  not hold that $(2w(u,v')+2w(v,u'))/2=4\leq w(u,v)=3$, so even if
  $u'$ and $v'$ are low-degree nodes, the degree of either $u'$ or
  $v'$ alone does not suffice to determine whether the contribution
  from an edge out of either $u$ or $v$ may be doubled.}
\label{fig:highdegreecounterex}
\end{figure}
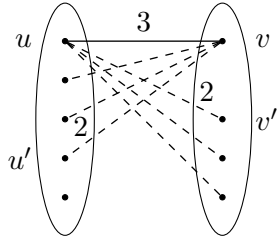

Figure~\ref{fig:highdegreecounterex} illustrates the base case of the
proof. This also shows why it is not possible to improve the bound by
low-degree considerations for the edges chosen for the lower
bound. For instance, it is not true that since $v'$ is not a
high-degree vertex, the contribution from edge $(u,v')$ can be counted
twice. This is only possible if \emph{all} vertices adjacent to $u$
are low degree. If we can keep an estimate of the maximum degree of
any adjacent vertex to $v$ for all free vertices $v\in F$ it is easy
to determine whether the lower bound contribution from high-degree
vertex $v$ can be multiplied by two, namely if this maximum degree is
smaller than $f_1$. The estimate can be the static maximum degrees in
$G$, and can be updated when a connected component contribution is
computed (see next section).

For each new subproblem, we can update the $T$ contribution in
$O(\deg(v))$ time steps, amortized over all vertices such that the
total time spent is $O(n+m)$ steps. To do this we maintain for each
free vertex of $(U_0,U_1)$, a) its \emph{free degree}, b) an index of
edges scanned so far, c) a count of seen (free) edges, and finally d)
the total weight of the seen edges. In total, four counts are
maintained per free vertex and for each subset $V_0\setminus U_0$ and
$V_1\setminus U_1$, making this relatively expensive in terms of space
needed per subproblem.

The free degree $\deg'(v)$ of vertex $v$ in $(U_0,U_1)$ is the number
of adjacent edges to free vertices (and is the degree of $v$ in the
induced subgraph). Initially, the free degree of vertex $v$ is just
its degree; the free degree is decreased each time a neighbor of $v$
is assigned to a subset. The count of seen edges shall be maintained
as $\max(\deg'(v)-f_i+1,0)\geq 0$ for each free node $v$, and we
maintain also the sum of the weights of these free edges. Note that
the number of seen edges for some free node $v$ (and their weight) may
have to be updated both as a result of an edge $(v,u)$ out of $v$
becoming assigned, or by some other node becoming assigned to subset
$V_i$. The three remaining invariant properties of the counts are now
maintained as follows. When a branching vertex $u$ is assigned to a
subset, either of $f_i$ is decreased by one; thus, for each free
vertex one more vertex must be counted as seen, and this is
accomplished by scanning edges of the vertex until an edge whose
endpoint is not fixed is met. The weight of this edge is added to the
total weight of seen edges. This is eventually the lower bound
contribution of the vertex. For vertices whose free degree $\deg'(v)$
decrease (that is, free vertices adjacent to the branching vertex $u$)
there are two cases. If edge $(v,u)$ has already been seen, that is if
edge $(v,u)$ is indexed before the currently scanned edge of $v$, then
the weight of the edge can simply be subtracted from the total weight
of seen edges of $v$.  Since we know the position $i$ of vertex $u$ in
the adjacency list of $v$, we have to subtract if $i$ is smaller than
the number of scanned edges of $v$. If on the other hand the edge
$(v,u)$ has not yet been seen (and scanned), the last (highest
weighted) seen edge must be made unseen (and its weight subtracted
from the total weight of seen edges), which is done by scanning back
from the currently scanned edge until a free edge is found. In total
one forwards and at most one backwards scan is made per adjacency array.

For unweighted graphs, the lower bound can be computed easily. The
contribution from a (high-degree) vertex is simply
$\max(0,\deg'(v)-f_0+1)$. Currently, we have only implemented the
general case, which is more costly (by some constant factor) than the
special, unweighted case.

Since this strengthening only contributes a nontrivial lower bound
increase in the presence of high-degree vertices (relative to the
subproblem $(U_0,U_1)$), we only compute the contribution if the
maximum degree $\deg'(v)$ of any free vertex $v$ is larger than
$f_1$. We can maintain an approximation of this maximum degree, namely
the maximum degree from the parent subproblem, and use this to trigger
the lower bound computation.

Considering just the total number of free edges gives a trivial, even
weaker lower bound for the term $w(V_0\setminus U_0,V_1\setminus
U_1)$. Assume there are more than $f_0(f_0-1)/2+f_1(f_1-1)/2$ edges in
the subgraph of $G$ induced by the free nodes $F$ ($f_0$ and and $f_1$
being the number of free nodes to be assigned to the subsets $U_0$ and
$U_1$, respectively). The sum of the weights of the least weight such
edges is a lower bound, since at least that number of edges must be in
the cut of any partition of $F$ into subsets of sizes $f_0$ and
$f_1$. If the induced subgraph has such a large number edges, at least
one will be of high degree larger than $f_0-1$ and $f_1-1$, and thus the
high-degree bound described above will be at least as strong, since it
counts at least as many edges of at least the same weight.

\subsection{The lower bound: a large unassigned component}
\label{sec:largecomponent}

Assume there are no high-degree vertices in the sense discussed
above. We make the observation that if there is a $k$-connected
component in the graph induced by $F$ of size greater than $f_0$, then
at least $k$ edges will cross in any partition of the $F$ vertices
into subsets of size $f_0$ and $f_1$ where here $f_0\geq f_1$. The sum
of the weights of the $k$ lightest edges in such a $k$-connected
component will thus be a lower bound. More generally, the
(unconstrained) minimum cut of the singly connected component of size
greater than $f_0$ will be a lower bound on the size constrained
cut. Let $C(U_0,U_1)$ be either the weight of the $k$ smallest edges
in a $k$-connected component of size greater than $f_0$ in the
subgraph induced by $F$, or the value of a minimum cut in such a
singly connected component.

\begin{proposition}
  \label{prop:component}
  For any given subproblem $(U_0,U_1)$ it holds that
  \begin{eqnarray*}
    C(U_0,U_1) & \leq & w(V_0\setminus U_0,V_1\setminus U_1)
  \end{eqnarray*}
  for any completion $(V_0,V_1)$.
\end{proposition}

\begin{proof}
  If there is a $k$-connected component larger than $f_0$ then there
  will be some nodes of this component in either subset of any partition
  of $F$ into subsets of sizes $f_0$ and $f_1$ (assuming $f_0\geq f_1$).
  At least $k$ edges of the large $k$-connected component must cross
  between $V_0\setminus U_0$ and $V_1\setminus U_1$. The sum of the weights
  of the lightest such $k$ edges are therefore a lower bound on the
  value of any cut. Likewise is any size-unconstrained minimum cut in
  the subgraph induced by the large component a lower bound.
\end{proof}

On the other hand, if there is a high-degree vertex in the sense
explained in the last section, then there is a connected component of
size at least $f_0+1$. Since the high-degree bound counts the
weights of particular edges (out of the high-degree vertices), this
bound is at least as strong as the connected components bound.

Here we settle for only a contribution of one lightest edge of a
large, connected component. The component is computed by a simple
breadth-first search traversal in $O(f+m)$ time steps, which can be
expensive compared to the other lower bound contributions.  The
computation is therefore triggered by maintaining an approximate size
of the largest component. Again, this approximation is just the size
of the largest component from the parent subproblem, and is updated
when a connected components computation is performed. During the graph
traversal we also update the maximum adjacent degree for each free
vertex. When a branching vertex is assigned, only the component to
which this vertex belongs can be affected. It would therefore be
possible to redo the connected component computation for these
affected vertices, and this could perhaps be of advantage for the
implementation, although not in the worst case; here, suitable dynamic
connected components algorithms would have to be used. This
optimization is not implemented currently, since it would require
extra arrays for maintaining component numbers and sizes and storing
the smallest edge weight for the components. Using DFS~\cite{Tarjan72}
for the graph traversal, it would be possible to compute the
2-connected components also in linear time.

Even if there is no large, $k$-connected component in the induced
subgraph, there may still be a lower bound contribution. For instance,
if the smallest $k$-connected component is larger than $f_1$, the
smaller of the vertex sets, then \emph{some} connected component will
cross between $U_0\setminus V_0$ and $U_1\setminus V_1$; thus, the sum
of the weights of the $k$ least cost edges over all $k$-connected
components will be a lower bound on the cut value; as above the least
(unconstrained) minimum cut would also be a lower bound. More
generally, if it is \emph{not} possible to pack a subset of the
connected components into a subset of size $f_1$, then at least one of
the components must have vertices on both sides of any cut. Again, the
smallest weight edge in the induced subgraph will be a lower bound on
the minimum cut value, and so will the minimum of all
size-unconstrained minimum cuts. Determining this contribution implies
determining that a corresponding subset sum problem does not have a
solution.  The subset sum problem is in itself
NP-hard~\cite{GareyJohnson79}, but might be small and special enough
that it could make sense to attempt a solution~\cite{Pisinger99}. We
have not pursued this idea further in our current implementation.

A possibly stronger bound is achieved by actually computing a minimum
cut value in the graph induced by the large, singly connected
component. An easily implementable algorithm~\cite{StoerWagner97} runs
in $O(mn+n^2\log n)$ time steps; a better, randomized
algorithm~\cite{Karger00} in $O(m\log^3 n)$ steps.  We have also not
yet experimented with this strengthening of the lower bound.

\subsection{Maintaining an upper bound}

The balancing step for the lower bound of the $w(V_0\setminus
U_0,V_1)+w(U_0,V_1\setminus U_1)$ term determines an explicit
assignment of the free vertices to the two subsets, that is a specific
completion. We can use this completion as an upper bound on the best
possible solution for the subproblem $(U_0,U_1)$.  Computing the cut
value of this partition would take $O(n+m)$ time, and might be too
expensive to do repeatedly.  However, for the case where the a large
connected component bound may apply, this computation could be done
almost for free. Furthermore, when branching on a vertex and creating
new subproblems, it is easy to determine which vertices will change in
the forced completions of the new subproblems. If only few vertices
change, the cut value of the completion can be updated more cheaply
similarly to what is done for instance in the Lin-Kernighan
heuristic~\cite{FiducciaMattheyses82,KernighanLin70,Traff06:kpartition}.

This computed upper bound is a solution candidate. Also, when upper and
lower bounds meet, the lower bound is tight for the subproblem
$(U_0,U_1)$, and no further branching is needed. Note, that this means
that there are no edges in the cut $(V_0\setminus U_0,V_1\setminus
U_1)$ for the particular completion $(V_0,V_1)$ induced by the
rebalancing lower bound.  We have not implemented this potential
improvement so far.

\subsection{Completion and branching rules}

A subproblem is essentially solved if either $|U_0|=s_0$ or
$|U_1|=s_1$: all free vertices can be assigned to the other subset.
Furthermore, if only one vertex is missing from, e.g., $U_0$, then the
vertex which has the smallest $D_1[v]$ value to the other subset plus
the smallest sum of free edges (which will all cross the cut) can be
assigned to $V_0$, and the remaining free vertices to subset $V_1$. No
other assignment can lead to a smaller cut value of the completion
$(V_0,V_1)$.

Another completion rule follows from the observation that the
completion implied by the lower bound with rebalancing
(Proposition~\ref{prop:rebalancing}) is an optimal solution if all free
vertices have free degree zero. This can easily be checked, and the
corresponding solution generated; this is also implemented, and led to
a reduction of a few (tens of) subproblems to be explored; since this
comes at virtually no cost (it requires only maintaining the number of
degree zero free vertices), this check and completion is always done
when rebalancing is enabled. More generally, if it can be inferred
that in this completion, there are no edges between sets $V_0\setminus
U_0$ and $V_1\setminus U_1$, the completion is optimal.

Other observations allows to reduce the worst-case number of
subproblems that needs to be generated. We state two such
observations:
\begin{enumerate}
\item
  For unassigned vertices with no free edges, not all possible
  assignments need to be checked. In particular, if there are $n$ such
  vertices, only $n+1$ of the possible $2^n$ assignments can lead to
  an completion value. The $i$th such subproblem for $i=0,\ldots n$
  would assign the $i$ nodes with the smallest value of
  $\delta(v)=D_1[v_i]-D_0[v_i]$ to $V_0$, and the remaining $n-i$ nodes
  to $V_1$. Since these nodes have no free edges, the only
  contribution to the cut can come from the edges to the assigned
  vertices in $U_0$ or $U_1$ as counted in $D_1[v]$ and
  $D_0[v]$. Swapping a vertex thus assigned to $V_0$ would lead to a
  larger cut value.
\item
  If there is a free edge between two nodes $u$ and $v$ each with degree
  one, the contribution to the cut of a completion is determined by
  $w(u,v)$ and the weight of the edges to the assigned vertices in $U_0$
  and $U_1$. Only one of the subproblems $(U_0\cup\{u\},U_1\cup\{v\})$
  or $(U_0\cup\{v\},U_1\cup\{u\})$ can lead to an optimal completion,
  namely the one with the smallest $D_0[u]+D_1[v]$ or $D_0[v]+D_1[u]$
  value. Therefore, only 3 instead of 4 possible subproblems must be
  generated.
\item
  In general, for a $k$-clique only $k+1$ instead of $2^k$ subproblems
  needs to be generated.
\end{enumerate}
These (and other, similar) observations can be used as
$n$-way branching rules, instead of the binary branching rule that
just generates two subproblems by assigning the chosen branching
vertex to either $U_0$ or $U_1$.

In our experiments, none of the first two rules above gave an
advantage, and were often detrimental in that too many subproblems
were generated too early. Thus, the benefit, if any, is not clear at
the moment, and such branching rules have not been considered further
here.

\section{Solving the weighted graph bipartitioning problem}

We use the task-parallel Pheet C++ framework as a general framework to
implement branch-and-bound algorithms. This is described extensively
in~\cite{Wimmer14:diss}, and briefly
in~\cite{Traff13:strategies,Traff14:priosched}. The basic idea is to
represent subproblems as tasks that can be executed in parallel when
enough have been created, and let the framework take care of the
selection of tasks in a priority-respecting order. To this end, Pheet
supports \emph{scheduling strategies} where tasks can be spawned with
an associated priority. A Pheet branch-and-bound task is shown in
Figure~\ref{fig:bbtask}. When a task is processed, it is first checked
that the task's lower bound is still smaller than the currently best,
feasible solution (a better solution could have been found between
the time the task was spawned and the time it is being processed).
The computed branching vertex is used to split the subproblem into two
(or more, but this is not shown here) new subproblems. For either, it
is checked whether it can already be completed, and in that case
whether it has lead to a new, better, global solution. If not, a new
task is spawned with some computed priority. The Pheet framework
will ensure that the subproblem is eventually processed by some
available processing unit, preferably in good (but possibly relaxed)
priority order.

\begin{figure}
\begin{lstlisting}[mathescape=true,columns=flexible]
template <class Pheet, 
template <class P, class SubProblem> class Logic, 
template <class P, class SubProblem> class SchedulingStrategy, size_t MaxSize>
void StrategyBBGraphBipartitioningTask<Pheet, Logic, SchedulingStrategy, MaxSize>::
operator()() {
	if(sub_problem->get_lower_bound() >= sub_problem->get_global_upper_bound()) {
		pc.num_irrelevant_tasks.incr();
		return;
	}

	SubProblem* sub_problem2 =
			sub_problem->split(pc.subproblem_pc);

	if(sub_problem->can_complete(pc.subproblem_pc)) {
		sub_problem->complete_solution(pc.subproblem_pc);
		sub_problem->update_solution(best, pc.subproblem_pc);
	}
	else if(sub_problem->get_lower_bound() < sub_problem->get_global_upper_bound()) {
		Pheet::template
			spawn_prio<Self>(strategy(sub_problem),
				sub_problem, best, pc);
		sub_problem = NULL;
	}

	if(sub_problem2->can_complete(pc.subproblem_pc)) {
		sub_problem2->complete_solution(pc.subproblem_pc);
		sub_problem2->update_solution(best, pc.subproblem_pc);
		delete sub_problem2;
	}
	else if(sub_problem2->get_lower_bound() < sub_problem2->get_global_upper_bound()) {
		Pheet::template
			spawn_prio<Self>(strategy(sub_problem2),
				sub_problem2, best, pc);
	}
	else {
		delete sub_problem2;
	}
}
\end{lstlisting}
\caption{A Pheet branch-and-bound task with binary branching.}
\label{fig:bbtask}
\end{figure}

\subsection{Initial subproblems and upper bound}

Branch-and-bound algorithms can benefit immensely from having a good
initial feasible solution or upper bound. In our current
implementation we use a simple, greedy strategy (corresponding to one
iteration of the minimum cut algorithm in~\cite{StoerWagner97}) to
produce an initial solution. Using a standard heuristic package like
METIS~\cite{karypis11} or SCOTCH~\cite{ChevalierPellegrini08} would be
a natural possibility to get a strong, initial upper bound. Easily
computable solutions as provided by variations of the Lin-Kernighan
heuristic are another
possibility~\cite{KernighanLin70,Traff06:kpartition}.

\subsection{Choosing good subproblems in parallel}

Branch-and-bound normally consider subproblems in some prioritized
order, with problems that are likely to lead to an improved solution
or to being cut off being preferred to other problems. The subproblem
priority order can have a large influence on the concrete performance
of the branch-and-bound procedure, even if there is no worst-case
difference. Here, we prioritize in by the difference lower bound and
an estimated upper bound, such that subproblems that are close to
their upper bound will be processed early. Other possibilities might
be worthwhile to explore.

The Pheet framework supports the possibility of prioritizing tasks and
processes tasks in (relaxed) priority order. Pheet relies on various,
relaxed, concurrent priority queues for this, which can provide
certain semantic and performance guarantees. We refer
to~\cite{Wimmer14:diss} for definitions of such semantics as well as
algorithmic and implementation details.

\subsection{Branching rules}

For the choice of branching vertex for each subproblem there a
likewise many possibilities, and the choice of branching rule can
likewise have a large effect on practical performance. In our current
implementation we branch on the vertex that will lead to the estimated
largest increase in the lower bound when assigned to either of the
subsets.

\section{Benchmark results}

We now present a selection of benchmark results for solving graph
bipartitioning problems using the Pheet framework with the various
lower bound contributions developed in the previous sections.

The experiments reported here are for simple, Erd\"os-R\'enyi random
graphs as in~\cite{Wimmer14:diss}. The graphs have $n$ nodes, and
edges are chosen with a given, uniform probability. The bounds and the
framework should be tested with standard test instances, for instance
those used
in~\cite{DellingGoldbergRazenshteynWerneck11,DellingGoldbergRazenshteynWerneck12,DellingWerneck12}. The
graphs are either weighted, in which case edge weights are chosen
uniformly at random with $w\in [1,1000]$; or unweighted, which we
achieve by choosing weights $w\in [1,1]$. We give results for sparser
graphs with edge probability $0.1$ (which is, asymptotically, of
course rather dense), medium dense graphs with edge probability $0.5$,
and dense graphs with edge probability $0.75$, and finally complete
graphs with edge probability $1$.

\subsection{Lower bound contributions}

We first investigate the difference between the different lower bound
contributions described in Section~\ref{sec:basicrebalance},
Section~\ref{sec:highdegree} and Section~\ref{sec:largecomponent}. To
do this, we solve the benchmark problems sequentially, using a
depth-first (non-prioritized) order on the generated subproblems. We
record the time to solution and relate that to the number of
subproblems that were explored.  Starting from the trivial lower bound
we track the reduction in number of subproblems and hopefully
proportional reduction in running time by gradually strengthening the
lower bound by adding the rebalancing contribution
(Proposition~\ref{prop:rebalancing}), the high-degree vertex
contribution (Proposition~\ref{prop:highdegree}), and the connected
components contribution (Proposition~\ref{prop:component}).

We give results from 5 differently generated random graphs (in Pheet
with seeds $0,1,2,3,4$) from each of the four categories. In addition
to the total time to solution and the number of explored subproblems,
we also give the number of times a new solution was found, and the
time at which the optimal (last) solution was found. For the cases
where the optimal solution is found late, a better, initial solution
could be of help. To check this, we also ran the experiments using the
optimal solution as initial solution; this gives an objective count of
how many problems must be explored to prove optimality, and is thus
indicative of the strength (or weakness) of the lower bound (for the
given branching rule; a different branching rule could change this;
the experiment is not sensitive to the choice of subproblem priority);
we only did this experiment for the strongest version of the lower
bound, and here we did not measure the actual running time; we just
list the, in most cases, smaller number of explored subproblems.

The sequential experiments were carried out on an Intel-based desktop
computer with a 4-core 3.4GHz Intel i7-2600 processor. We used
\texttt{gcc 4.7.2} under Debian 4.7.2-5 Linux.  All running times in
seconds, and the times recorded here for a single run only; on an
unloaded desktop the running times appear rather stable. The running
times are only indicative; whereas the various subproblem counts are
deterministic and exactly reproducible.

\subsubsection{Sparse graphs}

\begin{table}
\begin{center}
\begin{tabular}{|rrr|r|r|rr|r|}
\hline
$n$ & Prob. & $\max w$ & Time & Cut & Solutions & Subproblems & Opt.\ Time \\
\hline
40 & 0.1 & 1000 & 0.027267 & 7770 & 7 & 24662 & 0.026160 \\
40 & 0.1 & 1000 & 0.065237 & 9439 & 13 & 69439 & 0.063928 \\
40 & 0.1 & 1000 & 0.010008 & 5121 & 5 & 12520 & 0.009737 \\
40 & 0.1 & 1000 & 0.013096 & 6523 & 10 & 16627 & 0.012922 \\
40 & 0.1 & 1000 & 0.007524 & 6883 & 6 & 9655 & 0.005304 \\
50 & 0.1 & 1000 & 0.320722 & 12829 & 18 & 380379 & 0.313434 \\
50 & 0.1 & 1000 & 0.85727 & 14461 & 24 & 1012481 & 0.834905 \\
50 & 0.1 & 1000 & 0.450444 & 9096 & 23 & 529688 & 0.433826 \\
50 & 0.1 & 1000 & 0.343937 & 10150 & 20 & 391779 & 0.325539 \\
50 & 0.1 & 1000 & 0.862478 & 12438 & 20 & 942418 & 0.854184 \\
60 & 0.1 & 1000 & 12.0744 & 17502 & 29 & 11686971 & 11.278630 \\
60 & 0.1 & 1000 & 32.5406 & 22283 & 24 & 29141096 & 30.874791 \\
60 & 0.1 & 1000 & 6.74624 & 14585 & 20 & 6517009 & 6.500285 \\
60 & 0.1 & 1000 & 16.389 & 14794 & 26 & 15644183 & 16.296062 \\
60 & 0.1 & 1000 & 24.7185 & 20752 & 38 & 22683390 & 21.663690 \\
\hline
40 & 0.1 & 1 & 0.036729 & 19 & 5 & 45891 & 0.015383 \\
40 & 0.1 & 1 & 0.044857 & 24 & 2 & 57338 & 0.041048 \\
40 & 0.1 & 1 & 0.005785 & 14 & 1 & 8223 & 0.000006 \\
40 & 0.1 & 1 & 0.075284 & 17 & 12 & 93339 & 0.074613 \\
40 & 0.1 & 1 & 0.018821 & 17 & 4 & 23674 & 0.011489 \\
50 & 0.1 & 1 & 0.73198 & 28 & 6 & 781139 & 0.711416 \\
50 & 0.1 & 1 & 2.83857 & 32 & 10 & 3056976 & 2.755318 \\
50 & 0.1 & 1 & 0.515393 & 25 & 2 & 626003 & 0.419959 \\
50 & 0.1 & 1 & 1.20943 & 25 & 11 & 1273991 & 1.096330 \\
50 & 0.1 & 1 & 4.80033 & 33 & 14 & 4904437 & 4.631338 \\
60 & 0.1 & 1 & 21.8007 & 42 & 6 & 19637124 & 16.885930 \\
60 & 0.1 & 1 & 97.601 & 53 & 10 & 84588305 & 75.646062 \\
60 & 0.1 & 1 & 32.8037 & 39 & 12 & 29165329 & 22.103448 \\
60 & 0.1 & 1 & 9.48105 & 35 & 5 & 8931547 & 9.305966 \\
60 & 0.1 & 1 & 126.377 & 50 & 11 & 113479717 & 112.269199 \\
\hline
\end{tabular}
\end{center}
\caption{Sparse random graphs with edge probability $0.1$,
  $w\in[1,1000]$ and $w\in[1,1]$. Trivial lower bound.}
\label{tab:sparse-norebal}
\end{table}

\begin{table}
\begin{center}
\begin{tabular}{|rrr|r|r|rr|r|}
\hline
$n$ & Prob. & $\max w$ & Time & Cut & Solutions & Subproblems & Opt.\ Time \\
\hline
40 & 0.1 & 1000 & 0.007509 & 7770 & 7 & 3164 & 0.005886 \\
40 & 0.1 & 1000 & 0.006289 & 9439 & 13 & 2888 & 0.004836 \\
40 & 0.1 & 1000 & 0.001719 & 5121 & 5 & 735 & 0.001296 \\
40 & 0.1 & 1000 & 0.003596 & 6523 & 10 & 1843 & 0.003427 \\
40 & 0.1 & 1000 & 0.002714 & 6883 & 6 & 1178 & 0.000321 \\
50 & 0.1 & 1000 & 0.057015 & 12829 & 18 & 26314 & 0.052750 \\
50 & 0.1 & 1000 & 0.078241 & 14461 & 24 & 45226 & 0.071687 \\
50 & 0.1 & 1000 & 0.03234 & 9096 & 23 & 20536 & 0.021565 \\
50 & 0.1 & 1000 & 0.031847 & 10150 & 20 & 19785 & 0.017051 \\
50 & 0.1 & 1000 & 0.039503 & 12438 & 20 & 23618 & 0.038913 \\
60 & 0.1 & 1000 & 0.160308 & 17502 & 29 & 77425 & 0.053616 \\
60 & 0.1 & 1000 & 0.771585 & 22283 & 24 & 394574 & 0.421817 \\
60 & 0.1 & 1000 & 0.160184 & 14585 & 20 & 77410 & 0.065352 \\
60 & 0.1 & 1000 & 0.327008 & 14794 & 26 & 164106 & 0.273994 \\
60 & 0.1 & 1000 & 1.23954 & 20752 & 38 & 641762 & 0.544033 \\
\hline
40 & 0.1 & 1 & 0.004571 & 19 & 5 & 3678 & 0.000330 \\
40 & 0.1 & 1 & 0.004976 & 24 & 2 & 4116 & 0.003736 \\
40 & 0.1 & 1 & 0.001199 & 14 & 1 & 899 & 0.000010 \\
40 & 0.1 & 1 & 0.004194 & 17 & 12 & 3482 & 0.003519 \\
40 & 0.1 & 1 & 0.002543 & 17 & 4 & 1965 & 0.000275 \\
50 & 0.1 & 1 & 0.04406 & 28 & 6 & 30093 & 0.041470 \\
50 & 0.1 & 1 & 0.096952 & 32 & 10 & 67991 & 0.074959 \\
50 & 0.1 & 1 & 0.088082 & 25 & 2 & 60791 & 0.044789 \\
50 & 0.1 & 1 & 0.037535 & 25 & 11 & 25168 & 0.008754 \\
50 & 0.1 & 1 & 0.103074 & 33 & 14 & 70069 & 0.068494 \\
60 & 0.1 & 1 & 0.236067 & 42 & 6 & 128962 & 0.029440 \\
60 & 0.1 & 1 & 3.96175 & 53 & 10 & 2410554 & 1.395450 \\
60 & 0.1 & 1 & 0.990479 & 39 & 12 & 558587 & 0.075503 \\
60 & 0.1 & 1 & 0.251598 & 35 & 5 & 139853 & 0.209421 \\
60 & 0.1 & 1 & 4.33112 & 50 & 11 & 2599774 & 2.396716 \\
\hline
\end{tabular}
\end{center}
\caption{Sparse random graphs with edge probability $0.1$,
  $w\in[1,1000]$ and $w\in[1,1]$. Lower bound with rebalancing contribution.}
\label{tab:sparse-rebal}
\end{table}

\begin{table}
\begin{center}
\begin{tabular}{|rrr|r|r|rr|r|}
\hline
$n$ & Prob. & $\max w$ & Time & Cut & Solutions & Subproblems & Opt.\ Time \\
\hline
40 & 0.1 & 1000 & 0.004721 & 7770 & 7 & 3154 & 0.003703 \\
40 & 0.1 & 1000 & 0.004272 & 9439 & 13 & 2881 & 0.003260 \\
40 & 0.1 & 1000 & 0.001172 & 5121 & 5 & 724 & 0.000881 \\
40 & 0.1 & 1000 & 0.00273 & 6523 & 10 & 1837 & 0.002567 \\
40 & 0.1 & 1000 & 0.001854 & 6883 & 6 & 1178 & 0.000302 \\
50 & 0.1 & 1000 & 0.048337 & 12829 & 17 & 26293 & 0.044498 \\
50 & 0.1 & 1000 & 0.080309 & 14461 & 24 & 45197 & 0.073267 \\
50 & 0.1 & 1000 & 0.035161 & 9096 & 23 & 20489 & 0.023535 \\
50 & 0.1 & 1000 & 0.03498 & 10150 & 19 & 19772 & 0.018809 \\
50 & 0.1 & 1000 & 0.043134 & 12438 & 20 & 23592 & 0.042481 \\
60 & 0.1 & 1000 & 0.17328 & 17502 & 29 & 77404 & 0.058439 \\
60 & 0.1 & 1000 & 0.842746 & 22283 & 24 & 394561 & 0.458983 \\
60 & 0.1 & 1000 & 0.174728 & 14585 & 20 & 77410 & 0.071872 \\
60 & 0.1 & 1000 & 0.3576 & 14794 & 23 & 163998 & 0.299697 \\
60 & 0.1 & 1000 & 1.35194 & 20752 & 38 & 641762 & 0.594848 \\
\hline
40 & 0.1 & 1 & 0.005176 & 19 & 5 & 3673 & 0.000348 \\
40 & 0.1 & 1 & 0.005727 & 24 & 2 & 4113 & 0.004282 \\
40 & 0.1 & 1 & 0.00137 & 14 & 1 & 899 & 0.000006 \\
40 & 0.1 & 1 & 0.004799 & 17 & 12 & 3453 & 0.004016 \\
40 & 0.1 & 1 & 0.002827 & 17 & 4 & 1965 & 0.000286 \\
50 & 0.1 & 1 & 0.049363 & 28 & 6 & 30063 & 0.046491 \\
50 & 0.1 & 1 & 0.109406 & 32 & 10 & 67970 & 0.084869 \\
50 & 0.1 & 1 & 0.098004 & 25 & 2 & 60791 & 0.049769 \\
50 & 0.1 & 1 & 0.041763 & 25 & 11 & 25163 & 0.010042 \\
50 & 0.1 & 1 & 0.113649 & 33 & 14 & 70034 & 0.075498 \\
60 & 0.1 & 1 & 0.261038 & 42 & 6 & 128950 & 0.032600 \\
60 & 0.1 & 1 & 4.3952 & 53 & 10 & 2410551 & 1.544418 \\
60 & 0.1 & 1 & 1.08553 & 39 & 12 & 558587 & 0.084374 \\
60 & 0.1 & 1 & 0.280103 & 35 & 5 & 139844 & 0.233175 \\
60 & 0.1 & 1 & 4.77637 & 50 & 11 & 2599761 & 2.651952 \\
\hline
\end{tabular}
\end{center}
\caption{Sparse random graphs with edge probability $0.1$,
  $w\in[1,1000]$ and $w\in[1,1]$. Lower bound with rebalancing and 
  high-degree contributions.}
\label{tab:sparse-highdegree}
\end{table}

\begin{table}
\begin{center}
\begin{tabular}{|rrr|r|r|rrr|r|}
\hline
$n$ & Prob. & $\max w$ & Time & Cut & Solutions & Subproblems & With
optimal & Opt.\ Time \\
\hline
40 & 0.1 & 1000 & 0.007153 & 7770 & 7 & 3121 & 1701 & 0.005560 \\
40 & 0.1 & 1000 & 0.006443 & 9439 & 13 & 2842 & 1228 & 0.004901 \\
40 & 0.1 & 1000 & 0.001643 & 5121 & 5 & 719 & 280 & 0.001208 \\
40 & 0.1 & 1000 & 0.003775 & 6523 & 10 & 1837 & 1260 & 0.003540 \\
40 & 0.1 & 1000 & 0.002586 & 6883 & 6 & 1144 & 1051 & 0.000390 \\
50 & 0.1 & 1000 & 0.077288 & 12829 & 17 & 25850 & 13812 & 0.070880 \\
50 & 0.1 & 1000 & 0.1288 & 14461 & 24 & 44522 & 24324 & 0.117323 \\
50 & 0.1 & 1000 & 0.051628 & 9096 & 23 & 20021 & 10235 & 0.033532 \\
50 & 0.1 & 1000 & 0.054827 & 10150 & 19 & 19490 & 11494 & 0.028375 \\
50 & 0.1 & 1000 & 0.069346 & 12438 & 20 & 23157 & 7214 & 0.068327 \\
60 & 0.1 & 1000 & 0.288285 & 17502 & 29 & 75860 & 51678 & 0.091989 \\
60 & 0.1 & 1000 & 1.47682 & 22283 & 24 & 392842 & 250649 & 0.796601 \\
60 & 0.1 & 1000 & 0.285672 & 14585 & 20 & 76669 & 47630 & 0.115380 \\
60 & 0.1 & 1000 & 0.600668 & 14794 & 23 & 162745 & 66431 & 0.501793 \\
60 & 0.1 & 1000 & 2.26215 & 20752 & 38 & 627491 & 426709 & 0.984418 \\
\hline
40 & 0.1 & 1 & 0.005906 & 19 & 5 & 2907 & 2669 & 0.000434 \\
40 & 0.1 & 1 & 0.006806 & 24 & 2 & 3242 & 2218 & 0.005147 \\
40 & 0.1 & 1 & 0.001549 & 14 & 1 & 866 & 866 & 0.000008 \\
40 & 0.1 & 1 & 0.005656 & 17 & 12 & 3195 & 981 & 0.004748 \\
40 & 0.1 & 1 & 0.00333 & 17 & 4 & 1654 & 1524 & 0.000328 \\
50 & 0.1 & 1 & 0.059948 & 28 & 6 & 22192 & 13202 & 0.056434 \\
50 & 0.1 & 1 & 0.134392 & 32 & 10 & 49507 & 15853 & 0.103304 \\
50 & 0.1 & 1 & 0.114136 & 25 & 2 & 47341 & 37054 & 0.057939 \\
50 & 0.1 & 1 & 0.050109 & 25 & 11 & 18543 & 13614 & 0.011498 \\
50 & 0.1 & 1 & 0.135859 & 33 & 14 & 50585 & 24836 & 0.089800 \\
60 & 0.1 & 1 & 0.348262 & 42 & 6 & 96073 & 84417 & 0.040723 \\
60 & 0.1 & 1 & 5.77289 & 53 & 10 & 1728577 & 1526818 & 2.003938 \\
60 & 0.1 & 1 & 1.31056 & 39 & 12 & 404049 & 366201 & 0.095397 \\
60 & 0.1 & 1 & 0.354385 & 35 & 5 & 100256 & 42995 & 0.294146 \\
60 & 0.1 & 1 & 6.1523 & 50 & 11 & 1873524 & 1172829 & 3.360075 \\
\hline
\end{tabular}
\end{center}
\caption{Sparse random graphs with edge probability $0.1$,
  $w\in[1,1000]$ and $w\in[1,1]$. Lower bound with rebalancing, 
  high-degree and large connected component contributions.}
\label{tab:sparse-all}
\end{table}

The results for sparse graphs are given in
Table~\ref{tab:sparse-norebal}, Table~\ref{tab:sparse-rebal},
Table~\ref{tab:sparse-highdegree} and Table~\ref{tab:sparse-all}.  We
first notice (and this observation holds also for the other graph
categories) that the rebalancing lower bound contribution leads to a
huge reduction in number of subproblems; for the sparse graphs often
more than a factor of 20, and both for weighted and unweighted
problems. A similar reduction in running times follows. The
high-degree bound has, as would be expected, no effect here, the
number of subproblems is for all graphs the same. Fortunately, running
times seem to increase only slightly, which could mean that the
larger memory space needed to represent the subproblems for this bound
contribution is not in itself too costly. Adding the large connected
components contribution reduces the number of subproblems that have to
be considered by a significant factor less than 2, especially for the
unweighted graphs (as could be hoped for). Unfortunately, the extra
cost for repeatedly computing connected components outweigh the
reduction in number of subproblems, resulting in an increase in time
to solution by a small factor less than 2. As can be seen in
Table~\ref{tab:sparse-all}, the optimal solution is often found late,
about half-way through, and knowing the optimal solution as expected
leads to a significant reduction in numbers of subproblems that must
be explored; the reduction is less than a factor of 2, though.

\subsubsection{Medium dense graphs}

\begin{table}
\begin{center}
\begin{tabular}{|rrr|r|r|rr|r|}
\hline
$n$ & Prob. & $\max w$ & Time & Cut & Solutions & Subproblems & Opt.\ Time \\
\hline
35 & 0.5 & 1000 & 1.52957 & 58764 & 8 & 1631423 & 1.442220 \\
35 & 0.5 & 1000 & 1.45441 & 53759 & 13 & 1569317 & 1.090717 \\
35 & 0.5 & 1000 & 2.07381 & 57403 & 5 & 2281700 & 0.949366 \\
35 & 0.5 & 1000 & 1.00532 & 52375 & 12 & 1072568 & 0.479008 \\
35 & 0.5 & 1000 & 1.35941 & 51263 & 6 & 1542516 & 0.805244 \\
40 & 0.5 & 1000 & 8.13641 & 77452 & 7 & 7398960 & 7.119544 \\
40 & 0.5 & 1000 & 4.19626 & 65643 & 10 & 3916431 & 2.272946 \\
40 & 0.5 & 1000 & 7.92397 & 74034 & 17 & 7399920 & 6.992345 \\
40 & 0.5 & 1000 & 9.02878 & 69479 & 20 & 8579493 & 6.311949 \\
40 & 0.5 & 1000 & 3.27055 & 64952 & 8 & 2982317 & 2.088845 \\
45 & 0.5 & 1000 & 212.693 & 97409 & 13 & 173104481 & 109.589662 \\
45 & 0.5 & 1000 & 207.349 & 88328 & 15 & 182233817 & 171.094763 \\
45 & 0.5 & 1000 & 289.602 & 99578 & 9 & 244147116 & 254.704806 \\
45 & 0.5 & 1000 & 159.806 & 87290 & 24 & 136813957 & 115.479522 \\
45 & 0.5 & 1000 & 98.8471 & 82358 & 12 & 84298381 & 54.957056 \\
50 & 0.5 & 1000 & 1677.05 & 122708 & 13 & 1244044666 & 1270.327979 \\
50 & 0.5 & 1000 & 757.332 & 113679 & 13 & 568874993 & 31.799608 \\
50 & 0.5 & 1000 & 1654.28 & 119443 & 26 & 1190261679 & 1008.096581 \\
50 & 0.5 & 1000 & 1034.73 & 110783 & 10 & 785507028 & 886.124273 \\
50 & 0.5 & 1000 & 666.77 & 106336 & 14 & 509060437 & 446.841404 \\
\hline
35 & 0.5 & 1 & 4.65021 & 124 & 6 & 5269158 & 3.204704 \\
35 & 0.5 & 1 & 4.21001 & 120 & 7 & 4742440 & 1.624278 \\
35 & 0.5 & 1 & 4.2815 & 122 & 4 & 4667670 & 0.311393 \\
35 & 0.5 & 1 & 3.33108 & 118 & 4 & 3663708 & 1.444637 \\
35 & 0.5 & 1 & 4.15371 & 116 & 7 & 4665842 & 1.641493 \\
40 & 0.5 & 1 & 30.91 & 164 & 4 & 29116782 & 16.022385 \\
40 & 0.5 & 1 & 21.1244 & 152 & 2 & 20240985 & 19.215322 \\
40 & 0.5 & 1 & 25.9468 & 164 & 5 & 24385212 & 14.307037 \\
40 & 0.5 & 1 & 25.9925 & 155 & 5 & 24769899 & 10.099783 \\
40 & 0.5 & 1 & 19.4168 & 148 & 6 & 18709834 & 9.208803 \\
45 & 0.5 & 1 & 848.796 & 208 & 11 & 727360976 & 364.433880 \\
45 & 0.5 & 1 & 567.632 & 194 & 5 & 512035561 & 187.807958 \\
45 & 0.5 & 1 & 840.502 & 212 & 6 & 729022051 & 556.588036 \\
45 & 0.5 & 1 & 475.761 & 196 & 3 & 413970887 & 62.586686 \\
45 & 0.5 & 1 & 539.03 & 189 & 10 & 475797670 & 122.222540 \\
50 & 0.5 & 1 & 6641.34 & 259 & 4 & 826628810 & 2926.981140 \\
50 & 0.5 & 1 & 7660.52 & 253 & 11 & 1771570933 & 5053.982181 \\
50 & 0.5 & 1 & 6319.32 & 263 & 8 & 471230838 & 2461.749443 \\
50 & 0.5 & 1 & 5152.95 & 244 & 9 & -202499879 & 3026.881798 \\
50 & 0.5 & 1 & 3591.88 & 234 & 14 & -1543394689 & 2380.251020 \\
\hline
\end{tabular}
\end{center}
\caption{Medium random graphs with edge probability $0.5$,
  $w\in[1,1000]$ and $w\in[1,1]$. Trivial lower bound.}
\label{tab:medium-norebal}
\end{table}

\begin{table}
\begin{center}
\begin{tabular}{|rrr|r|r|rr|r|}
\hline
$n$ & Prob. & $\max w$ & Time & Cut & Solutions & Subproblems & Opt.\ Time \\
\hline
35 & 0.5 & 1000 & 0.264013 & 58764 & 8 & 240742 & 0.260441 \\
35 & 0.5 & 1000 & 0.248807 & 53759 & 13 & 247377 & 0.203659 \\
35 & 0.5 & 1000 & 0.376443 & 57403 & 5 & 385777 & 0.184042 \\
35 & 0.5 & 1000 & 0.168552 & 52375 & 12 & 157670 & 0.077293 \\
35 & 0.5 & 1000 & 0.273458 & 51263 & 6 & 270768 & 0.173758 \\
40 & 0.5 & 1000 & 0.954789 & 77452 & 7 & 804652 & 0.830155 \\
40 & 0.5 & 1000 & 0.499157 & 65643 & 10 & 407117 & 0.194435 \\
40 & 0.5 & 1000 & 0.863879 & 74034 & 17 & 745162 & 0.729009 \\
40 & 0.5 & 1000 & 1.32082 & 69479 & 20 & 1124210 & 0.848951 \\
40 & 0.5 & 1000 & 0.210091 & 64952 & 8 & 164562 & 0.091707 \\
45 & 0.5 & 1000 & 16.2536 & 97409 & 13 & 12920447 & 9.249073 \\
45 & 0.5 & 1000 & 20.8354 & 88328 & 15 & 17088208 & 18.617669 \\
45 & 0.5 & 1000 & 21.8805 & 99578 & 9 & 17823261 & 20.776830 \\
45 & 0.5 & 1000 & 14.1073 & 87290 & 24 & 11424168 & 11.470726 \\
45 & 0.5 & 1000 & 6.7862 & 82358 & 12 & 5273280 & 4.230401 \\
50 & 0.5 & 1000 & 92.208 & 122708 & 13 & 69114119 & 64.386391 \\
50 & 0.5 & 1000 & 39.5426 & 113679 & 13 & 28419256 & 0.508509 \\
50 & 0.5 & 1000 & 83.9424 & 119443 & 26 & 63410193 & 42.537509 \\
50 & 0.5 & 1000 & 60.9577 & 110783 & 10 & 44579202 & 49.544491 \\
50 & 0.5 & 1000 & 40.4662 & 106336 & 14 & 29260832 & 24.309945 \\
\hline
35 & 0.5 & 1 & 0.84897 & 124 & 6 & 941536 & 0.646175 \\
35 & 0.5 & 1 & 0.691073 & 120 & 7 & 749916 & 0.279613 \\
35 & 0.5 & 1 & 0.737433 & 122 & 4 & 799902 & 0.047950 \\
35 & 0.5 & 1 & 0.473472 & 118 & 4 & 512213 & 0.225050 \\
35 & 0.5 & 1 & 0.72023 & 116 & 7 & 804766 & 0.309161 \\
40 & 0.5 & 1 & 3.13725 & 164 & 4 & 3007647 & 1.411749 \\
40 & 0.5 & 1 & 2.19887 & 152 & 2 & 2049958 & 2.011402 \\
40 & 0.5 & 1 & 2.6435 & 164 & 5 & 2485163 & 1.317400 \\
40 & 0.5 & 1 & 3.08175 & 155 & 5 & 2912912 & 0.900051 \\
40 & 0.5 & 1 & 1.80619 & 148 & 6 & 1689634 & 0.698424 \\
45 & 0.5 & 1 & 68.1693 & 208 & 11 & 59670375 & 31.835521 \\
45 & 0.5 & 1 & 49.584 & 194 & 5 & 42974470 & 17.051052 \\
45 & 0.5 & 1 & 58.1044 & 212 & 6 & 50732290 & 42.331557 \\
45 & 0.5 & 1 & 33.5946 & 196 & 3 & 28382980 & 3.404407 \\
45 & 0.5 & 1 & 44.7034 & 189 & 10 & 38908212 & 9.043913 \\
50 & 0.5 & 1 & 418.594 & 259 & 4 & 344992577 & 158.152543 \\
50 & 0.5 & 1 & 518.243 & 253 & 11 & 431767620 & 317.159786 \\
50 & 0.5 & 1 & 347.785 & 263 & 8 & 284400088 & 120.231741 \\
50 & 0.5 & 1 & 271.259 & 244 & 9 & 215621591 & 150.366613 \\
50 & 0.5 & 1 & 189.26 & 234 & 14 & 148813817 & 122.718955 \\
\hline
\end{tabular}
\end{center}
\caption{Medium random graphs with edge probability $0.5$,
  $w\in[1,1000]$ and $w\in[1,1]$. Lower bound with rebalancing
  contribution.}
\label{tab:medium-rebal}
\end{table}

\begin{table}
\begin{center}
\begin{tabular}{|rrr|r|r|rr|r|}
\hline
$n$ & Prob. & $\max w$ & Time & Cut & Solutions & Subproblems & Opt.\ Time \\
\hline
35 & 0.5 & 1000 & 0.437692 & 58764 & 8 & 235292 & 0.430839 \\
35 & 0.5 & 1000 & 0.425497 & 53759 & 13 & 242591 & 0.347885 \\
35 & 0.5 & 1000 & 0.666342 & 57403 & 5 & 373692 & 0.325182 \\
35 & 0.5 & 1000 & 0.288164 & 52375 & 12 & 154324 & 0.131943 \\
35 & 0.5 & 1000 & 0.47135 & 51263 & 6 & 265518 & 0.301217 \\
40 & 0.5 & 1000 & 1.67566 & 77452 & 7 & 774893 & 1.457920 \\
40 & 0.5 & 1000 & 0.846506 & 65643 & 10 & 399509 & 0.331395 \\
40 & 0.5 & 1000 & 1.5265 & 74034 & 17 & 728341 & 1.290223 \\
40 & 0.5 & 1000 & 2.2194 & 69479 & 20 & 1100933 & 1.449646 \\
40 & 0.5 & 1000 & 0.323746 & 64952 & 8 & 163934 & 0.142388 \\
45 & 0.5 & 1000 & 29.0639 & 97409 & 13 & 12614199 & 16.531646 \\
45 & 0.5 & 1000 & 35.5419 & 88328 & 15 & 16843898 & 31.671609 \\
45 & 0.5 & 1000 & 38.5121 & 99578 & 9 & 17408144 & 36.496099 \\
45 & 0.5 & 1000 & 24.9732 & 87290 & 24 & 11281722 & 20.296342 \\
45 & 0.5 & 1000 & 11.2825 & 82358 & 12 & 5241779 & 7.054381 \\
50 & 0.5 & 1000 & 167.633 & 122708 & 13 & 66634939 & 117.179447 \\
50 & 0.5 & 1000 & 69.967 & 113679 & 13 & 27799777 & 0.975712 \\
50 & 0.5 & 1000 & 159.028 & 119443 & 26 & 60947964 & 80.593920 \\
50 & 0.5 & 1000 & 105.542 & 110783 & 10 & 43918527 & 86.097362 \\
50 & 0.5 & 1000 & 69.2837 & 106336 & 14 & 28457488 & 42.036457 \\
\hline
35 & 0.5 & 1 & 1.35922 & 124 & 6 & 871277 & 1.033276 \\
35 & 0.5 & 1 & 1.07911 & 120 & 7 & 716940 & 0.434534 \\
35 & 0.5 & 1 & 1.17949 & 122 & 4 & 725398 & 0.071981 \\
35 & 0.5 & 1 & 0.752047 & 118 & 4 & 493487 & 0.356672 \\
35 & 0.5 & 1 & 1.16846 & 116 & 7 & 770453 & 0.503175 \\
40 & 0.5 & 1 & 5.07885 & 164 & 4 & 2864374 & 2.286510 \\
40 & 0.5 & 1 & 3.47258 & 152 & 2 & 2018241 & 3.173854 \\
40 & 0.5 & 1 & 4.23372 & 164 & 5 & 2423006 & 2.115559 \\
40 & 0.5 & 1 & 4.81996 & 155 & 5 & 2795728 & 1.407476 \\
40 & 0.5 & 1 & 2.75878 & 148 & 6 & 1656079 & 1.090572 \\
45 & 0.5 & 1 & 109.491 & 208 & 11 & 56006759 & 50.681042 \\
45 & 0.5 & 1 & 78.2504 & 194 & 5 & 41889646 & 26.893112 \\
45 & 0.5 & 1 & 95.5201 & 212 & 6 & 48662675 & 69.489188 \\
45 & 0.5 & 1 & 52.6926 & 196 & 3 & 27582111 & 5.511826 \\
45 & 0.5 & 1 & 68.3469 & 189 & 10 & 37849057 & 13.994667 \\
50 & 0.5 & 1 & 679.514 & 259 & 4 & 319026234 & 254.246279 \\
50 & 0.5 & 1 & 828.198 & 253 & 11 & 415976171 & 507.750523 \\
50 & 0.5 & 1 & 564.802 & 263 & 8 & 256492804 & 192.467646 \\
50 & 0.5 & 1 & 411.546 & 244 & 9 & 211109755 & 227.832552 \\
50 & 0.5 & 1 & 293.403 & 234 & 14 & 142303065 & 189.947916 \\
\hline
\end{tabular}
\end{center}
\caption{Medium random graphs with edge probability $0.5$,
  $w\in[1,1000]$ and $w\in[1,1]$. Lower bound with rebalancing and
  high-degree contributions.}
\label{tab:medium-highdegree}
\end{table}

\begin{table}
\begin{center}
\begin{tabular}{|rrr|r|r|rrr|r|}
\hline
$n$ & Prob. & $\max w$ & Time & Cut & Solutions & Subproblems & With
optimal & Opt.\ Time \\
\hline
35 & 0.5 & 1000 & 0.457685 & 58764 & 8 & 235288 & 168998 & 0.450484 \\
35 & 0.5 & 1000 & 0.456458 & 53759 & 13 & 242548 & 186531 & 0.372767 \\
35 & 0.5 & 1000 & 0.701454 & 57403 & 5 & 373655 & 357593 & 0.347343 \\
35 & 0.5 & 1000 & 0.303973 & 52375 & 12 & 154306 & 107663 & 0.139753 \\
35 & 0.5 & 1000 & 0.490333 & 51263 & 6 & 265476 & 224962 & 0.315305 \\
40 & 0.5 & 1000 & 1.76439 & 77452 & 7 & 774877 & 602759 & 1.540306 \\
40 & 0.5 & 1000 & 0.856906 & 65643 & 10 & 399493 & 381139 & 0.338882 \\
40 & 0.5 & 1000 & 1.52563 & 74034 & 17 & 728331 & 660421 & 1.286866 \\
40 & 0.5 & 1000 & 2.32557 & 69479 & 20 & 1100861 & 750247 & 1.510679 \\
40 & 0.5 & 1000 & 0.366715 & 64952 & 8 & 163902 & 120719 & 0.161070 \\
45 & 0.5 & 1000 & 30.03 & 97409 & 13 & 12614160 & 9331432 & 17.022816 \\
45 & 0.5 & 1000 & 36.9243 & 88328 & 15 & 16843385 & 13513268 & 33.014433 \\
45 & 0.5 & 1000 & 40.0554 & 99578 & 9 & 17408058 & 15693005 & 38.014035 \\
45 & 0.5 & 1000 & 25.6186 & 87290 & 24 & 11281384 & 6465924 & 20.866948 \\
45 & 0.5 & 1000 & 12.2618 & 82358 & 12 & 5240191 & 3495765 & 7.735454 \\
50 & 0.5 & 1000 & 171.224 & 122708 & 13 & 66634689 & 55686051 & 119.794947 \\
50 & 0.5 & 1000 & 72.0477 & 113679 & 13 & 27799722 & 27585136 & 1.010014 \\
50 & 0.5 & 1000 & 160.152 & 119443 & 26 & 60947916 & 43152850 & 80.989391 \\
50 & 0.5 & 1000 & 106.719 & 110783 & 10 & 43917268 & 31788303 & 87.113578 \\
50 & 0.5 & 1000 & 70.4171 & 106336 & 14 & 28456665 & 22696796 & 42.707401 \\
\hline
35 & 0.5 & 1 & 1.4499 & 124 & 6 & 855438 & 601343 & 1.101491 \\
35 & 0.5 & 1 & 1.20145 & 120 & 7 & 696530 & 618260 & 0.486955 \\
35 & 0.5 & 1 & 1.23778 & 122 & 4 & 716422 & 696641 & 0.072427 \\
35 & 0.5 & 1 & 0.854104 & 118 & 4 & 479529 & 393921 & 0.407605 \\
35 & 0.5 & 1 & 1.29394 & 116 & 7 & 747361 & 645569 & 0.551096 \\
40 & 0.5 & 1 & 5.50524 & 164 & 4 & 2818789 & 2344215 & 2.458640 \\
40 & 0.5 & 1 & 4.05059 & 152 & 2 & 1959144 & 1691452 & 3.707572 \\
40 & 0.5 & 1 & 4.71051 & 164 & 5 & 2373464 & 2187064 & 2.329712 \\
40 & 0.5 & 1 & 5.4678 & 155 & 5 & 2720383 & 2432052 & 1.599556 \\
40 & 0.5 & 1 & 3.40263 & 148 & 6 & 1586829 & 1236544 & 1.291652 \\
45 & 0.5 & 1 & 118.224 & 208 & 11 & 55380489 & 46060711 & 55.248718 \\
45 & 0.5 & 1 & 91.6279 & 194 & 5 & 40993549 & 36002769 & 30.710736 \\
45 & 0.5 & 1 & 102.48 & 212 & 6 & 48161062 & 40186440 & 74.656017 \\
45 & 0.5 & 1 & 62.6245 & 196 & 3 & 26867057 & 25884153 & 6.233091 \\
45 & 0.5 & 1 & 87.8081 & 189 & 10 & 36532129 & 33329439 & 17.232088 \\
50 & 0.5 & 1 & 737.273 & 259 & 4 & 315805991 & 291043130 & 274.235916 \\
50 & 0.5 & 1 & 940.956 & 253 & 11 & 408914142 & 309539236 & 575.418819 \\
50 & 0.5 & 1 & 583.4 & 263 & 8 & 255530043 & 214170641 & 198.905050 \\
50 & 0.5 & 1 & 521.97 & 244 & 9 & 204611587 & 171705443 & 285.835977 \\
50 & 0.5 & 1 & 362.011 & 234 & 14 & 138985134 & 98499266 & 229.681082 \\
\hline
\end{tabular}
\end{center}
\caption{Medium random graphs with edge probability $0.5$,
  $w\in[1,1000]$ and $w\in[1,1]$. Lower bound with rebalancing,
  high-degree and large component contributions.}
\label{tab:medium-all}
\end{table}

The results for medium dense graphs are listed in
Table~\ref{tab:medium-norebal}, Table~\ref{tab:medium-rebal},
Table~\ref{tab:medium-highdegree} and Table~\ref{tab:medium-all}.
Again, the benefits from the rebalancing contribution are enormous,
both for weighted and unweighted case, and obviously pay off
proportionally in running time (factors of 15 and more). Here, the
high-degree contribution is triggered and leads to a small reduction
in numbers of subproblems, but the computation is expensive and has a
negative effect on the time to solution which can almost double. The
same holds for the large connected components contribution, which
although the number of subproblems can be reduced slightly, increases
the running times by a small factor less than 2. In most cases the
optimal solution is found relatively late, and there would therefore
be a benefit (in number of subproblems to explore) of having a better
initial solution; the effect is less than a factor of 2, though.

\subsubsection{Dense graphs}

\begin{table}
\begin{center}
\begin{tabular}{|rrr|r|r|rr|r|}
\hline
$n$ & Prob. & $\max w$ & Time & Cut & Solutions & Subproblems & Opt.\ Time \\
\hline
30 & 0.75 & 1000 & 0.370472 & 70916 & 11 & 360252 & 0.259478 \\
30 & 0.75 & 1000 & 0.337704 & 68761 & 2 & 349245 & 0.083185 \\
30 & 0.75 & 1000 & 0.221823 & 67895 & 2 & 214933 & 0.132072 \\
30 & 0.75 & 1000 & 0.308535 & 66801 & 7 & 313060 & 0.253589 \\
30 & 0.75 & 1000 & 0.309682 & 65323 & 11 & 317570 & 0.217466 \\
35 & 0.75 & 1000 & 13.593 & 101149 & 7 & 12514058 & 6.705426 \\
35 & 0.75 & 1000 & 6.52348 & 88464 & 8 & 5895137 & 4.358233 \\
35 & 0.75 & 1000 & 5.20966 & 91901 & 7 & 4572475 & 4.766952 \\
35 & 0.75 & 1000 & 4.94418 & 85501 & 11 & 4295028 & 3.135605 \\
35 & 0.75 & 1000 & 6.67847 & 87457 & 22 & 5989707 & 2.272913 \\
40 & 0.75 & 1000 & 93.225 & 131755 & 8 & 72199908 & 58.842574 \\
40 & 0.75 & 1000 & 65.6515 & 121963 & 10 & 51272133 & 14.000955 \\
40 & 0.75 & 1000 & 70.0444 & 126009 & 14 & 53027510 & 52.134243 \\
40 & 0.75 & 1000 & 53.4483 & 119664 & 9 & 41483027 & 44.666847 \\
40 & 0.75 & 1000 & 49.6401 & 114953 & 16 & 39378102 & 20.038079 \\
\hline
30 & 0.75 & 1 & 0.891637 & 147 & 3 & 979764 & 0.725699 \\
30 & 0.75 & 1 & 1.02601 & 150 & 4 & 1128185 & 0.109506 \\
30 & 0.75 & 1 & 0.836867 & 148 & 2 & 906552 & 0.797915 \\
30 & 0.75 & 1 & 0.938687 & 148 & 1 & 1011748 & 0.000007 \\
30 & 0.75 & 1 & 0.932776 & 145 & 5 & 1043643 & 0.221146 \\
35 & 0.75 & 1 & 32.3945 & 205 & 1 & 31235221 & 0.000009 \\
35 & 0.75 & 1 & 22.7523 & 195 & 4 & 21695745 & 12.457475 \\
35 & 0.75 & 1 & 37.0252 & 206 & 6 & 36168861 & 24.928672 \\
35 & 0.75 & 1 & 28.6312 & 198 & 5 & 27985743 & 11.113570 \\
35 & 0.75 & 1 & 22.8631 & 191 & 6 & 22339313 & 10.276111 \\
40 & 0.75 & 1 & 333.888 & 272 & 6 & 279093966 & 32.442125 \\
40 & 0.75 & 1 & 309.575 & 264 & 9 & 261737608 & 149.204937 \\
40 & 0.75 & 1 & 423.295 & 277 & 5 & 354920538 & 96.527282 \\
40 & 0.75 & 1 & 255.759 & 261 & 10 & 209231563 & 153.626048 \\
40 & 0.75 & 1 & 216.547 & 255 & 6 & 184462859 & 49.831940 \\
\hline
\end{tabular}
\end{center}
\caption{Dense random graphs with edge probability $0.75$,
  $w\in[1,1000]$ and $w\in[1,1]$. Trivial lower bound.}
\label{tab:dense-norebal}
\end{table}

\begin{table}
\begin{center}
\begin{tabular}{|rrr|r|r|rr|r|}
\hline
$n$ & Prob. & $\max w$ & Time & Cut & Solutions & Subproblems & Opt.\ Time \\
\hline
30 & 0.75 & 1000 & 0.103219 & 70916 & 11 & 96349 & 0.072895 \\
30 & 0.75 & 1000 & 0.08406 & 68761 & 2 & 94414 & 0.016650 \\
30 & 0.75 & 1000 & 0.042504 & 67895 & 2 & 47347 & 0.022612 \\
30 & 0.75 & 1000 & 0.06844 & 66801 & 7 & 79773 & 0.058002 \\
30 & 0.75 & 1000 & 0.072394 & 65323 & 11 & 82569 & 0.051860 \\
35 & 0.75 & 1000 & 2.48431 & 101149 & 7 & 2724215 & 1.309536 \\
35 & 0.75 & 1000 & 0.987641 & 88464 & 8 & 1011439 & 0.721924 \\
35 & 0.75 & 1000 & 0.538253 & 91901 & 7 & 534528 & 0.519065 \\
35 & 0.75 & 1000 & 0.58678 & 85501 & 11 & 574069 & 0.428499 \\
35 & 0.75 & 1000 & 1.16092 & 87457 & 22 & 1185220 & 0.420064 \\
40 & 0.75 & 1000 & 12.4906 & 131755 & 8 & 12185394 & 7.514862 \\
40 & 0.75 & 1000 & 6.68102 & 121963 & 10 & 6317423 & 1.201540 \\
40 & 0.75 & 1000 & 7.36096 & 126009 & 14 & 6943774 & 5.513143 \\
40 & 0.75 & 1000 & 4.93249 & 119664 & 9 & 4624870 & 3.989565 \\
40 & 0.75 & 1000 & 5.72612 & 114953 & 16 & 5422791 & 1.996018 \\
\hline
30 & 0.75 & 1 & 0.238 & 147 & 3 & 307003 & 0.191913 \\
30 & 0.75 & 1 & 0.301815 & 150 & 4 & 379298 & 0.034334 \\
30 & 0.75 & 1 & 0.189407 & 148 & 2 & 243934 & 0.181643 \\
30 & 0.75 & 1 & 0.250688 & 148 & 1 & 316154 & 0.000008 \\
30 & 0.75 & 1 & 0.258292 & 145 & 5 & 328164 & 0.063881 \\
35 & 0.75 & 1 & 4.93283 & 205 & 1 & 5974035 & 0.000011 \\
35 & 0.75 & 1 & 3.33951 & 195 & 4 & 3959624 & 1.939788 \\
35 & 0.75 & 1 & 6.65941 & 206 & 6 & 8060167 & 4.803945 \\
35 & 0.75 & 1 & 4.09707 & 198 & 5 & 4963393 & 1.857012 \\
35 & 0.75 & 1 & 3.98347 & 191 & 6 & 4768656 & 1.942667 \\
40 & 0.75 & 1 & 39.3355 & 272 & 6 & 43757388 & 3.069326 \\
40 & 0.75 & 1 & 41.3252 & 264 & 9 & 46667399 & 20.465780 \\
40 & 0.75 & 1 & 64.7553 & 277 & 5 & 72302392 & 14.290673 \\
40 & 0.75 & 1 & 32.1253 & 261 & 10 & 34193672 & 18.884349 \\
40 & 0.75 & 1 & 27.4485 & 255 & 6 & 29658475 & 5.525352 \\
\hline
\end{tabular}
\end{center}
\caption{Dense random graphs with edge probability $0.75$,
  $w\in[1,1000]$ and $w\in[1,1]$. Lower bound with rebalancing contribution.}
\label{tab:dense-rebal}
\end{table}

\begin{table}
\begin{center}
\begin{tabular}{|rrr|r|r|rr|r|}
\hline
$n$ & Prob. & $\max w$ & Time & Cut & Solutions & Subproblems & Opt.\ Time \\
\hline
30 & 0.75 & 1000 & 0.110197 & 70916 & 11 & 59014 & 0.073739 \\
30 & 0.75 & 1000 & 0.110183 & 68761 & 2 & 57776 & 0.021733 \\
30 & 0.75 & 1000 & 0.050028 & 67895 & 2 & 24208 & 0.026681 \\
30 & 0.75 & 1000 & 0.093968 & 66801 & 7 & 49346 & 0.079876 \\
30 & 0.75 & 1000 & 0.100189 & 65323 & 11 & 54630 & 0.070855 \\
35 & 0.75 & 1000 & 3.26961 & 101149 & 7 & 1572739 & 1.722546 \\
35 & 0.75 & 1000 & 1.4405 & 88464 & 8 & 642251 & 1.052328 \\
35 & 0.75 & 1000 & 0.809674 & 91901 & 7 & 338742 & 0.782488 \\
35 & 0.75 & 1000 & 0.845042 & 85501 & 11 & 357992 & 0.619300 \\
35 & 0.75 & 1000 & 1.58746 & 87457 & 22 & 704004 & 0.571071 \\
40 & 0.75 & 1000 & 14.2917 & 131755 & 8 & 5677755 & 8.680114 \\
40 & 0.75 & 1000 & 10.558 & 121963 & 10 & 4217589 & 1.862208 \\
40 & 0.75 & 1000 & 9.04176 & 126009 & 14 & 3509650 & 6.967428 \\
40 & 0.75 & 1000 & 7.21843 & 119664 & 9 & 2803812 & 5.846905 \\
40 & 0.75 & 1000 & 9.29349 & 114953 & 16 & 3709982 & 3.219045 \\
\hline
30 & 0.75 & 1 & 0.157407 & 147 & 3 & 92701 & 0.127816 \\
30 & 0.75 & 1 & 0.192106 & 150 & 4 & 112422 & 0.022864 \\
30 & 0.75 & 1 & 0.106529 & 148 & 2 & 60224 & 0.102453 \\
30 & 0.75 & 1 & 0.146673 & 148 & 1 & 82837 & 0.000009 \\
30 & 0.75 & 1 & 0.2163 & 145 & 5 & 128667 & 0.052001 \\
35 & 0.75 & 1 & 2.55448 & 205 & 1 & 1294630 & 0.000010 \\
35 & 0.75 & 1 & 1.64483 & 195 & 4 & 811987 & 1.005831 \\
35 & 0.75 & 1 & 4.28534 & 206 & 6 & 2258574 & 3.174467 \\
35 & 0.75 & 1 & 3.71848 & 198 & 5 & 1962477 & 1.797938 \\
35 & 0.75 & 1 & 2.18208 & 191 & 6 & 1090292 & 1.109113 \\
40 & 0.75 & 1 & 17.4109 & 272 & 6 & 7898119 & 1.560261 \\
40 & 0.75 & 1 & 26.1942 & 264 & 9 & 12358406 & 13.616799 \\
40 & 0.75 & 1 & 20.491 & 277 & 5 & 9320364 & 5.435163 \\
40 & 0.75 & 1 & 16.4077 & 261 & 10 & 7301820 & 9.935822 \\
40 & 0.75 & 1 & 18.6961 & 255 & 6 & 8532342 & 3.864666 \\
\hline
\end{tabular}
\end{center}
\caption{Dense random graphs with edge probability $0.75$,
  $w\in[1,1000]$ and $w\in[1,1]$. Lower bound with rebalancing and
  high-degree contributions.}
\label{tab:dense-highdegree}
\end{table}

\begin{table}
\begin{center}
\begin{tabular}{|rrr|r|r|rrr|r|}
\hline
$n$ & Prob. & $\max w$ & Time & Cut & Solutions & Subproblems & With
optimal & Opt.\ Time \\
\hline
30 & 0.75 & 1000 & 0.111838 & 70916 & 11 & 59014 & 49411 & 0.074850 \\
30 & 0.75 & 1000 & 0.111387 & 68761 & 2 & 57776 & 55645 & 0.022114 \\
30 & 0.75 & 1000 & 0.050347 & 67895 & 2 & 24208 & 23955 & 0.026551 \\
30 & 0.75 & 1000 & 0.095215 & 66801 & 7 & 49346 & 36983 & 0.080854 \\
30 & 0.75 & 1000 & 0.101543 & 65323 & 11 & 54630 & 36903 & 0.071837 \\
35 & 0.75 & 1000 & 3.28723 & 101149 & 7 & 1572739 & 1425159 & 1.734716 \\
35 & 0.75 & 1000 & 1.44365 & 88464 & 8 & 642251 & 549589 & 1.059350 \\
35 & 0.75 & 1000 & 0.803225 & 91901 & 7 & 338742 & 274701 & 0.776440 \\
35 & 0.75 & 1000 & 0.849026 & 85501 & 11 & 357992 & 294054 & 0.623855 \\
35 & 0.75 & 1000 & 1.58858 & 87457 & 22 & 704004 & 578057 & 0.566048 \\
40 & 0.75 & 1000 & 14.177 & 131755 & 8 & 5677755 & 4348704 & 8.611801 \\
40 & 0.75 & 1000 & 10.4853 & 121963 & 10 & 4217589 & 3966813 & 1.860220 \\
40 & 0.75 & 1000 & 8.99509 & 126009 & 14 & 3509650 & 2261887 & 6.931919 \\
40 & 0.75 & 1000 & 7.1699 & 119664 & 9 & 2803812 & 2452372 & 5.804552 \\
40 & 0.75 & 1000 & 9.26082 & 114953 & 16 & 3709982 & 3174225 & 3.209197 \\
\hline
30 & 0.75 & 1 & 0.157288 & 147 & 3 & 92700 & 72481 & 0.127690 \\
30 & 0.75 & 1 & 0.191835 & 150 & 4 & 112422 & 104662 & 0.022669 \\
30 & 0.75 & 1 & 0.105953 & 148 & 2 & 60224 & 34917 & 0.101923 \\
30 & 0.75 & 1 & 0.146986 & 148 & 1 & 82837 & 82837 & 0.000009 \\
30 & 0.75 & 1 & 0.215565 & 145 & 5 & 128666 & 109322 & 0.051429 \\
35 & 0.75 & 1 & 2.55151 & 205 & 1 & 1294624 & 1294624 & 0.000011 \\
35 & 0.75 & 1 & 1.63822 & 195 & 4 & 811984 & 695869 & 1.002095 \\
35 & 0.75 & 1 & 4.27547 & 206 & 6 & 2258569 & 1735175 & 3.168967 \\
35 & 0.75 & 1 & 3.71503 & 198 & 5 & 1962473 & 1535044 & 1.798483 \\
35 & 0.75 & 1 & 2.17589 & 191 & 6 & 1090289 & 916161 & 1.106713 \\
40 & 0.75 & 1 & 17.4431 & 272 & 6 & 7898116 & 7616903 & 1.562321 \\
40 & 0.75 & 1 & 26.1431 & 264 & 9 & 12358396 & 8928926 & 13.591346 \\
40 & 0.75 & 1 & 20.3156 & 277 & 5 & 9320364 & 8395805 & 5.417039 \\
40 & 0.75 & 1 & 16.2719 & 261 & 10 & 7301815 & 5643512 & 9.864531 \\
40 & 0.75 & 1 & 18.6772 & 255 & 6 & 8532341 & 7852227 & 3.847519 \\
\hline
\end{tabular}
\end{center}
\caption{Dense random graphs with edge probability $0.75$,
  $w\in[1,1000]$ and $w\in[1,1]$. Lower bound with rebalancing,
  high-degree and large connected component contributions.}
\label{tab:dense-all}
\end{table}

The results for the five dense graphs are shown in in
Table~\ref{tab:dense-norebal}, Table~\ref{tab:dense-rebal},
Table~\ref{tab:dense-highdegree} and Table~\ref{tab:dense-all}.  As
for the other graphs, the rebalancing contribution has the largest
effect, and is huge. The high-degree bound now gives a significant
reduction in number of subproblems, especially for the unweighted
problems where the reduction is large enough to lead to a worthwhile
reduction in running time. The large connected component contribution
is rarely triggered here, leads only to a very small change
in number of subproblems, and overall hardly affects the running time.
The reduction in number of subproblems when the optimal solution is
known initially is not as large as for the previous cases.

\subsubsection{Complete  graphs}

\begin{table}
\begin{center}
\begin{tabular}{|rrr|r|r|rr|r|}
\hline
$n$ & Prob. & $\max w$ & Time & Cut & Solutions & Subproblems & Opt.\ Time \\
\hline
20 & 1 & 1000 & 0.006859 & 44780 & 2 & 5911 & 0.000485 \\
20 & 1 & 1000 & 0.005944 & 40637 & 8 & 5308 & 0.001448 \\
20 & 1 & 1000 & 0.006428 & 44723 & 2 & 5745 & 0.005717 \\
20 & 1 & 1000 & 0.004846 & 41657 & 4 & 4204 & 0.001624 \\
20 & 1 & 1000 & 0.005349 & 40891 & 9 & 4967 & 0.004198 \\
30 & 1 & 1000 & 1.20566 & 99972 & 7 & 1074882 & 0.236397 \\
30 & 1 & 1000 & 1.14248 & 91583 & 9 & 1007547 & 1.114057 \\
30 & 1 & 1000 & 1.0778 & 96494 & 10 & 958071 & 0.500738 \\
30 & 1 & 1000 & 1.08817 & 96948 & 4 & 967051 & 0.008571 \\
30 & 1 & 1000 & 0.971279 & 93390 & 12 & 843119 & 0.814758 \\
\hline
20 & 1 & 1 & 0.01453 & 100 & 1 & 24309 & 0.000004 \\
20 & 1 & 1 & 0.014483 & 100 & 1 & 24309 & 0.000003 \\
20 & 1 & 1 & 0.014618 & 100 & 1 & 24309 & 0.000004 \\
20 & 1 & 1 & 0.014633 & 100 & 1 & 24309 & 0.000003 \\
20 & 1 & 1 & 0.014174 & 100 & 1 & 24309 & 0.000003 \\
30 & 1 & 1 & 15.0829 & 225 & 1 & 20058299 & 0.000005 \\
30 & 1 & 1 & 15.3418 & 225 & 1 & 20058299 & 0.000006 \\
30 & 1 & 1 & 15.2103 & 225 & 1 & 20058299 & 0.000006 \\
30 & 1 & 1 & 15.3602 & 225 & 1 & 20058299 & 0.000006 \\
30 & 1 & 1 & 15.4689 & 225 & 1 & 20058299 & 0.000007 \\
\hline
\end{tabular}
\end{center}
\caption{Complete random graphs, $w\in[1,1000]$ and
  $w\in[1,1]$. Trivial lower bound.}
\label{tab:complete-norebal}
\end{table}

\begin{table}
\begin{center}
\begin{tabular}{|rrr|r|r|rr|r|}
\hline
$n$ & Prob. & $\max w$ & Time & Cut & Solutions & Subproblems & Opt.\ Time \\
\hline
20 & 1 & 1000 & 0.005071 & 44780 & 2 & 3754 & 0.000456 \\
20 & 1 & 1000 & 0.003911 & 40637 & 8 & 3032 & 0.001306 \\
20 & 1 & 1000 & 0.004503 & 44723 & 2 & 3580 & 0.004025 \\
20 & 1 & 1000 & 0.002315 & 41657 & 4 & 1985 & 0.000856 \\
20 & 1 & 1000 & 0.003182 & 40891 & 9 & 2663 & 0.002738 \\
30 & 1 & 1000 & 0.335414 & 99972 & 7 & 356804 & 0.074178 \\
30 & 1 & 1000 & 0.324173 & 91583 & 9 & 344801 & 0.317206 \\
30 & 1 & 1000 & 0.259961 & 96494 & 10 & 283930 & 0.135704 \\
30 & 1 & 1000 & 0.306426 & 96948 & 4 & 334555 & 0.004207 \\
30 & 1 & 1000 & 0.241482 & 93390 & 12 & 251313 & 0.207057 \\
\hline
20 & 1 & 1 & 0.01766 & 100 & 1 & 24309 & 0.000005 \\
20 & 1 & 1 & 0.017638 & 100 & 1 & 24309 & 0.000004 \\
20 & 1 & 1 & 0.017909 & 100 & 1 & 24309 & 0.000004 \\
20 & 1 & 1 & 0.017622 & 100 & 1 & 24309 & 0.000004 \\
20 & 1 & 1 & 0.017693 & 100 & 1 & 24309 & 0.000003 \\
30 & 1 & 1 & 17.8267 & 225 & 1 & 20058299 & 0.000008 \\
30 & 1 & 1 & 17.9244 & 225 & 1 & 20058299 & 0.000007 \\
30 & 1 & 1 & 17.8562 & 225 & 1 & 20058299 & 0.000007 \\
30 & 1 & 1 & 18.5786 & 225 & 1 & 20058299 & 0.000008 \\
30 & 1 & 1 & 17.7369 & 225 & 1 & 20058299 & 0.000008 \\
\hline
\end{tabular}
\end{center}
\caption{Complete random graphs, $w\in[1,1000]$ and $w\in[1,1]$. Lower
  bound with rebalancing contribution.}
\label{tab:complete-rebal}
\end{table}

\begin{table}
\begin{center}
\begin{tabular}{|rrr|r|r|rr|r|}
\hline
$n$ & Prob. & $\max w$ & Time & Cut & Solutions & Subproblems & Opt.\ Time \\
\hline
20 & 1 & 1000 & 0.001866 & 44780 & 2 & 1194 & 0.000202 \\
20 & 1 & 1000 & 0.001798 & 40637 & 8 & 1154 & 0.000523 \\
20 & 1 & 1000 & 0.001702 & 44723 & 2 & 1114 & 0.001546 \\
20 & 1 & 1000 & 0.000833 & 41657 & 4 & 529 & 0.000307 \\
20 & 1 & 1000 & 0.001809 & 40891 & 9 & 1250 & 0.001554 \\
30 & 1 & 1000 & 0.157174 & 99972 & 7 & 79146 & 0.029180 \\
30 & 1 & 1000 & 0.166044 & 91583 & 9 & 83938 & 0.162576 \\
30 & 1 & 1000 & 0.117455 & 96494 & 10 & 58498 & 0.064518 \\
30 & 1 & 1000 & 0.124764 & 96948 & 4 & 60501 & 0.001721 \\
30 & 1 & 1000 & 0.114907 & 93390 & 12 & 56732 & 0.098666 \\
\hline
20 & 1 & 1 & 7e-06 & 100 & 1 & 0 & 0.000006 \\
20 & 1 & 1 & 6e-06 & 100 & 1 & 0 & 0.000006 \\
20 & 1 & 1 & 6e-06 & 100 & 1 & 0 & 0.000004 \\
20 & 1 & 1 & 6e-06 & 100 & 1 & 0 & 0.000005 \\
20 & 1 & 1 & 6e-06 & 100 & 1 & 0 & 0.000005 \\
30 & 1 & 1 & 9e-06 & 225 & 1 & 0 & 0.000008 \\
30 & 1 & 1 & 9e-06 & 225 & 1 & 0 & 0.000008 \\
30 & 1 & 1 & 9e-06 & 225 & 1 & 0 & 0.000008 \\
30 & 1 & 1 & 9e-06 & 225 & 1 & 0 & 0.000008 \\
30 & 1 & 1 & 8e-06 & 225 & 1 & 0 & 0.000008 \\
\hline
\end{tabular}
\end{center}
\caption{Complete random graphs, $w\in[1,1000]$ and $w\in[1,1]$. Lower
  bound with rebalancing and high-degree contributions.}
\label{tab:complete-highdegree}
\end{table}

\begin{table}
\begin{center}
\begin{tabular}{|rrr|r|r|rrr|r|}
\hline
$n$ & Prob. & $\max w$ & Time & Cut & Solutions & Subproblems & With
optimal & Opt.\ Time \\
\hline
20 & 1 & 1000 & 0.001868 & 44780 & 2 & 1194 & 1143 & 0.000206 \\
20 & 1 & 1000 & 0.001751 & 40637 & 8 & 1154 & 1118 & 0.000518 \\
20 & 1 & 1000 & 0.001712 & 44723 & 2 & 1114 & 866 & 0.001556 \\
20 & 1 & 1000 & 0.000822 & 41657 & 4 & 529 & 471 & 0.000306 \\
20 & 1 & 1000 & 0.001801 & 40891 & 9 & 1250 & 925 & 0.001548 \\
30 & 1 & 1000 & 0.157706 & 99972 & 7 & 79146 & 76482 & 0.029056 \\
30 & 1 & 1000 & 0.165645 & 91583 & 9 & 83938 & 66955 & 0.162109 \\
30 & 1 & 1000 & 0.116554 & 96494 & 10 & 58498 & 41410 & 0.063725 \\
30 & 1 & 1000 & 0.124496 & 96948 & 4 & 60501 & 60086 & 0.001730 \\
30 & 1 & 1000 & 0.114147 & 93390 & 12 & 56732 & 44219 & 0.098063 \\
\hline
20 & 1 & 1 & 7e-06 & 100 & 1 & 0 & 0 & 0.000006 \\
20 & 1 & 1 & 6e-06 & 100 & 1 & 0 & 0 & 0.000005 \\
20 & 1 & 1 & 6e-06 & 100 & 1 & 0 & 0 & 0.000005 \\
20 & 1 & 1 & 5e-06 & 100 & 1 & 0 & 0 & 0.000004 \\
20 & 1 & 1 & 5e-06 & 100 & 1 & 0 & 0 & 0.000004 \\
30 & 1 & 1 & 1e-05 & 225 & 1 & 0 & 0 & 0.000009 \\
30 & 1 & 1 & 9e-06 & 225 & 1 & 0 & 0 & 0.000008 \\
30 & 1 & 1 & 1e-05 & 225 & 1 & 0 & 0 & 0.000009 \\
30 & 1 & 1 & 9e-06 & 225 & 1 & 0 & 0 & 0.000009 \\
30 & 1 & 1 & 9e-06 & 225 & 1 & 0 & 0 & 0.000008 \\
\hline
\end{tabular}
\end{center}
\caption{Complete random graphs, $w\in[1,1000]$ and $w\in[1,1]$. Lower
  bound with rebalancing, high-degree and large connected component
  contributions.}
\label{tab:complete-all}
\end{table}

Results for the complete graphs can be found in
Table~\ref{tab:complete-norebal}, Table~\ref{tab:complete-rebal},
Table~\ref{tab:complete-highdegree} and Table~\ref{tab:complete-all}.
Here, the rebalancing contribution is much smaller, and only for
weighted graphs; but as can be expected the high-degree contribution
can instead be significant. Indeed, for the unweighted graphs, this
leads to a bound which immediately proves that the initial, heuristic
solution is optimal, and a reduction in number of subproblems from
20058299 to 0. The large connected components contribution is of
course not triggered.

\subsection{Parallel computing aspects}

To illustrate that the Pheet framework can efficiently distribute the
branch-and-bound search over a (large) number of cores, we include
results for the parallel solution of some of the graph problems from
the previous sections using now a prioritized search with some of the
priority data structures implemented in Pheet. For details, see
again~\cite{Wimmer14:diss}, and
also~\cite{Traff13:priosched,Traff13:stratcorr}. The Pheet framework
with the branch-and-bound code and the lower bounds developed in this
report can be downloaded from \url{www.pheet.org}.

The parallel graph partitioning experiments were performed on an
80-core Intel system with 1TB of memory consisting of eight 10-core
Xeon E7-8850 processors. Experiments were run under Debian Linux and
the framework compiled with \texttt{gcc 4.9.1}.

The plots in Table~\ref{fig:sparse-60} to Table~\ref{fig:large-40}
illustrate the speed-ups that can be achieved with increasing number
of cores. The reported running times in seconds are the averages of 30
repeated runs with one graph type.

\begin{figure}
  \includegraphics[width=\textwidth]{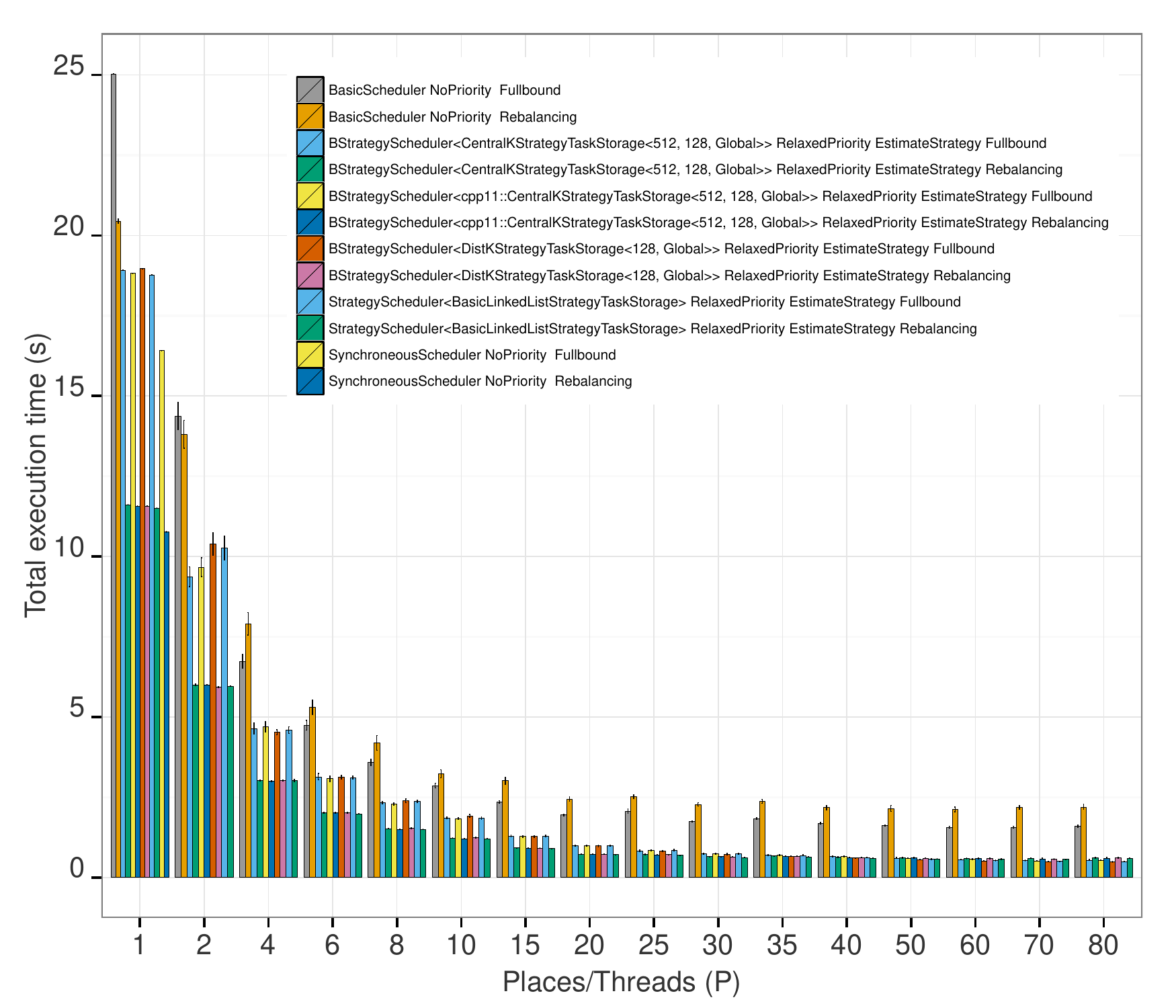}
\caption{Scalability for a sparse graph, $n=60$, edge probability
  $0.1$, $w\in [1,1000]$ with different scheduling strategies with
  1 to 80 cores.}
\label{fig:sparse-60}
\end{figure}

\begin{figure}
  \includegraphics[width=\textwidth]{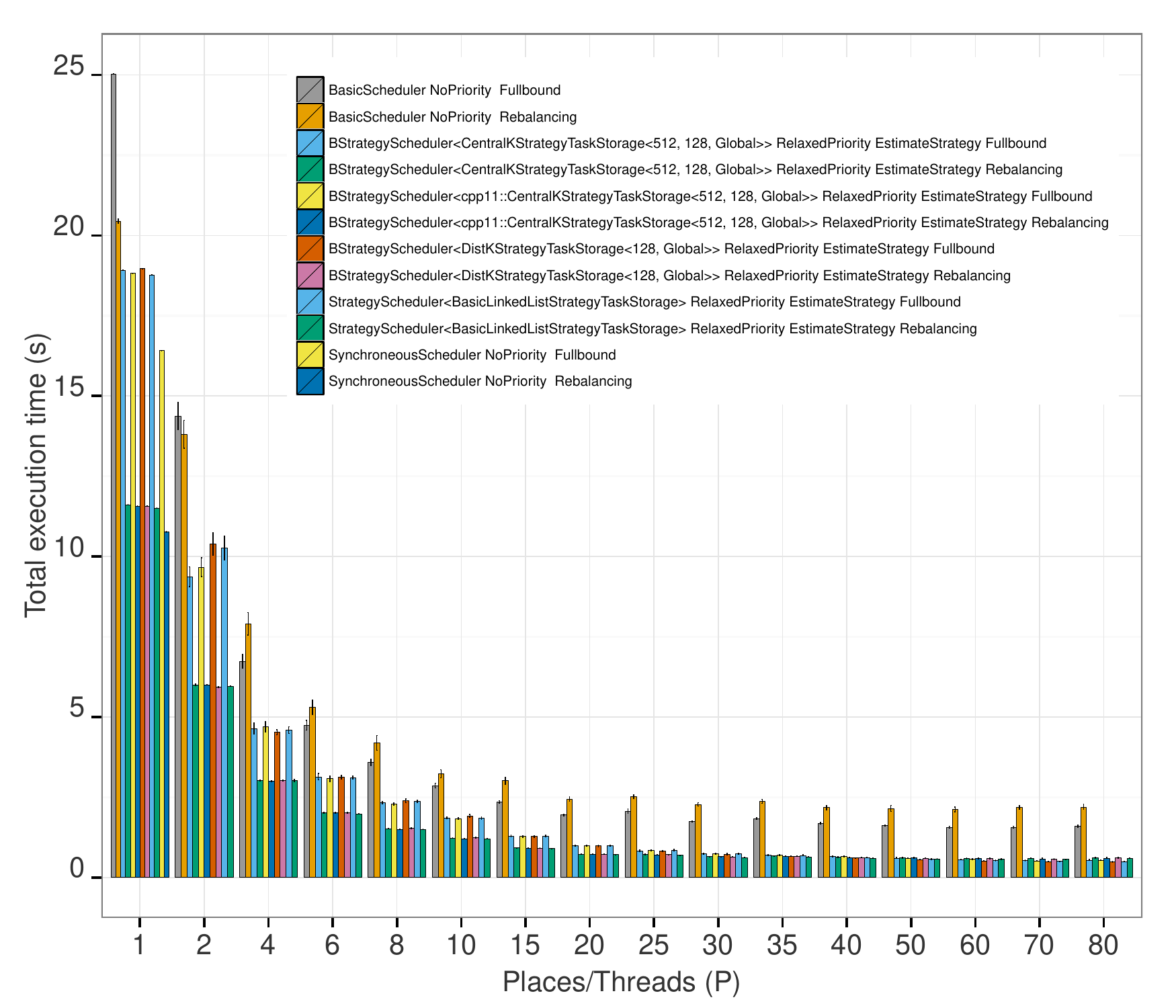}
\caption{Scalability for a medium dense graph, $n=45$, edge probability
  $0.5$, $w\in [1,1000]$ with different scheduling strategies with
  1 to 80 cores.}
\label{fig:medium-45}
\end{figure}

\begin{figure}
  \includegraphics[width=\textwidth]{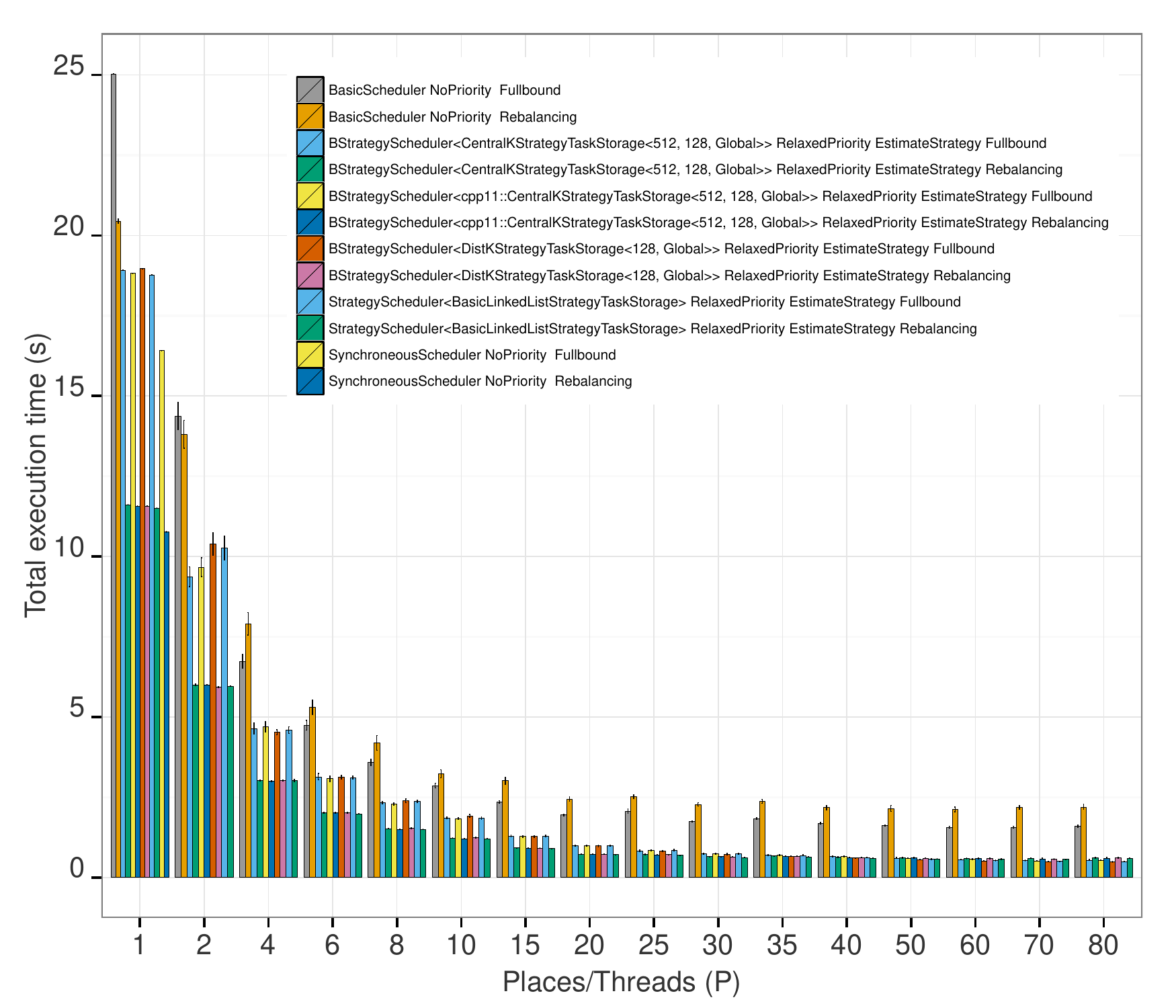}
\caption{Scalability for a dense graph, $n=40$, edge probability
  $0.75$, $w\in [1,1000]$ with different scheduling strategies
  with 1 to 80 cores.}
\label{fig:large-40}
\end{figure}

Scheduling strategies make it possible to prioritize tasks
representing graph partitioning subproblems, and select the most
promising task for processing. Most promising can mean either the
globally best task, the locally best task, or the task that is
globally best according to a relaxed correctness
criterion~\cite{Wimmer14:diss}. In the experiments a basic
work-stealing scheduler (legend ``BasicScheduler'' and ``NoPriority'')
not supporting priorities was compared against schedulers supporting
strategies and priority queues with relaxed semantics (legend
``BStrategyScheduler'' and ``RelaxedPriority''). The rebalancing lower
bound (legend ``Rebalancing'') is compared against the full bound with
rebalancing, high-degree and connected-components contributions
(legend ``Fullbound'').

As can be seen in the three concrete cases, running times decreases
with increasing number of cores, up till at least half the machine (40
cores). Prioritizing tasks provide significant reductions in running
time. It is also interesting that the full bound, which in the
sequential setting was often more expensive than the rebalancing bound
becomes cheaper than the rebalancing bound as the number of cores
increase (after four cores).

\section{Concluding remarks}

The purpose of this note was to resurrect and improve an old,
combinatorial lower bound for the weighted graph partitioning problem,
and to use this lower bound together with a modern, parallel
task-scheduling framework for solving weighted graph bipartitioning
problems as fast as possible. The results presented here a
preliminary, and a number of possible improvements were discussed. The
challenge to see whether the bound and the framework is competitive
with current state-of-the art (combinatorial) approaches for the exact
solution of graph partitioning problems (for certain types of graphs)
remains.

\bibliographystyle{abbrv}
\bibliography{traff,parallel}

\end{document}